\documentclass[letterpaper,UKenglish,clevered,autoref]{lipics-v2019}

\usepackage{graphicx}
\usepackage{balance}
\usepackage{algorithm}
\usepackage[noend]{algpseudocode}
\usepackage{url}
\usepackage{hyphenat}
\usepackage[font=small,labelfont=bf]{caption}
\usepackage{tikz,pgfplots,pgfplotstable}

\nolinenumbers

\pgfplotsset{compat=1.14}

\usetikzlibrary{arrows}
\usetikzlibrary{positioning}
\usetikzlibrary{patterns}

\newtheorem{thm}{Theorem}

\newtheorem{defn}{Definition}
\newtheorem{lem}{Lemma}
\newtheorem{cor}{Corollary}
\newtheorem{rem}{Remark}

\newtheorem{clm}{Claim}
\newtheorem{prob}{Problem}

\title{Reverse Prevention Sampling for Misinformation Mitigation in Social Networks\footnote{A preliminary version of this work will appear in \emph{ICDT 2020}.}}

\author{Michael~Simpson}{Department of Computer Science, University of Victoria, Canada}{}{}{}

\author{Venkatesh~Srinivasan}{Department of Computer Science, University of Victoria, Canada}{}{}{}

\author{Alex~Thomo}{Department of Computer Science, University of Victoria, Canada}{\{simpsonm,srinivas,thomo\}@uvic.ca}{}{}

\authorrunning{M. Simpson and V. Srinivasan and A. Thomo}

\Copyright{Michael Simpson and Venkatesh Srinivasan and Alex Thomo}



\begin{document}

\maketitle

\begin{abstract}
\sloppy In this work, we consider misinformation propagating through a social network and study the problem of its prevention. In this problem, a ``bad'' campaign starts propagating from a set of seed nodes in the network and we use the notion of a limiting (or ``good'') campaign to counteract the effect of misinformation. The goal is to identify a set of $k$ users that need to be convinced to adopt the limiting campaign so as to minimize the number of people that adopt the ``bad'' campaign at the end of both propagation processes.

This work presents \emph{RPS} (Reverse Prevention Sampling), an algorithm that provides a scalable solution to the misinformation mitigation problem. Our theoretical analysis shows that \emph{RPS} runs in $O((k + l)(n + m)(\frac{1}{1 - \gamma}) \log n / \epsilon^2 )$ expected time and returns a $(1 - 1/e - \epsilon)$-approximate solution with at least $1 - n^{-l}$ probability (where $\gamma$ is a typically small network parameter and $l$ is a confidence parameter). The time complexity of \emph{RPS} substantially improves upon the previously best-known algorithms that run in time $\Omega(m n k \cdot POLY(\epsilon^{-1}))$. We experimentally evaluate \emph{RPS} on large datasets and show that it outperforms the state-of-the-art solution by several orders of magnitude in terms of running time. This demonstrates that misinformation mitigation can be made practical while still offering strong theoretical guarantees.
\end{abstract}

\section{Introduction}

Social networks allow for widespread distribution of knowledge and information in modern society as they have rapidly become a place to hear the news and discuss social topics. Information can spread quickly through the network, eventually reaching a large audience, especially so for influential users. While the ease of information propagation in social networks can be beneficial, it can also have disruptive effects. In recent years, the number of high profile instances of misinformation causing severe real-world effects has risen sharply. These examples range across a number of social media platforms and topics \cite{telegraph2018_foster, independent2018_oppenheim, telegraph2018_graham, cnet2018_hautala, guardian2018_solon, aljazeera2018_abeshouse}. For example, a series of bogus tweets from a trusted news network referring to explosions at the White House caused immediate and extensive repercussions in the financial markets \cite{telegraph2018_foster}. During a recent shooting at YouTube's headquarters, and before police managed to secure the area, a wave of misinformation and erroneous accusations were widely disseminated on Twitter causing panic and confusion \cite{independent2018_oppenheim, telegraph2018_graham}. Finally, there has been much discussion on the role misinformation and fake news played in the 2016 U.S.\ presidential election with sites such as Reddit and Facebook being accused of harbouring and spreading divisive content and misinformation \cite{cnet2018_hautala, guardian2018_solon, aljazeera2018_abeshouse}. 
Thus, in order for social networks to serve as a reliable platform for disseminating critical information, it is necessary to have tools to limit the spread of misinformation.

Budak et al.\ \cite{budak2011limiting} were among the first to formulate the problem of misinformation mitigation as a combinatorial optimization problem. By building upon the seminal work of Kempe et al.\ \cite{kempe2003} on \emph{influence maximization} to a model that can handle multiple campaigns (``bad'' and ``good''), they present a greedy approach that provides a $(1 - 1/e - \epsilon)$-approximate solution. Unfortunately, the greedy approach of \cite{budak2011limiting} is plagued by the same scaling issues as \cite{kempe2003} when considering large social networks and is further exacerbated by the added complexity of tracking multiple cascades which requires costly shortest path computations. This leads us to the motivating question for this paper: Can we find scalable algorithms for the misinformation mitigation problem introduced in \cite{budak2011limiting}?

The scalability hurdle in the single campaign setting was recently resolved by Borgs et al.\ \cite{borgs2012} when they made a theoretical breakthrough that fundamentally shifts the way in which we view the influence maximization problem. Their key insight was to reverse the question of ``what subset of the network can a particular user influence'' to ``who could have influenced a particular user''. Their sampling method runs in close to linear time and returns a $(1 - 1/e - \epsilon)$-approximate solution with at least $1 - n^{-l}$ probability. 
In addition, Tang et al.\ \cite{Tang2014} presented a significant advance that improved the practical efficiency of Borgs et al. through a careful theoretical analysis that rids their approach of a large hidden constant in the runtime guarantee. Borgs et al.~\cite{borgs2012} leave open the question whether their framework can be extended to other influence propagation models. 

In this work, we resolve the question of \cite{borgs2012} for the misinformation mitigation problem and achieve scalability in the multi-campaign model. We complement our theoretical analysis with extensive experiments which show an improvement of several orders of magnitude over Budak et al.~\cite{budak2011limiting}. Since influence in the single campaign setting corresponds to reachability in the network, our solution requires mapping the concept of reachability to an analogous notion in the multi-campaign model for misinformation mitigation. Our first contribution is to show that reachability alone is not sufficient in determining the ability to save a particular node from the bad campaign. In order to address this challenge, we introduce a crucial notion of ``obstructed'' nodes, which are nodes such that all paths leading to them can be blocked by the bad campaign.

Using our newly defined notion of obstruction, we develop an efficient algorithm for the misinformation mitigation problem that provides much improved scalability over the existing Monte Carlo-based greedy approach of \cite{budak2011limiting}. A novel component of this algorithm is a procedure to compute the set of unobstructed nodes that could have saved a particular node from adopting the misinformation. We obtain theoretical guarantees on the expected runtime and solution quality for our new approach and show that its expected runtime substantially improves upon the expected runtime of \cite{budak2011limiting}. Additionally, we rule out sublinear algorithms for our problem through a lower bound on the time required to obtain a constant approximation.

Finally, from an experimental point of view, we show that our algorithm gives a significant improvement over the state of the art algorithm and can efficiently handle graphs with more than 50 million edges. In summary, the contributions of this paper are:

\vspace{-5pt}

\begin{enumerate}
\item \sloppy We introduce the concept of \emph{obstructed} nodes that fully captures the necessary conditions for preventing the adoption of misinformation in the multi-campaign model. In the process, we close a gap in the work of \cite{budak2011limiting}.
\item We design and implement a novel procedure for computing the set of nodes that could save a particular user from adopting the misinformation.
\item We propose a misinformation mitigation approach that returns a $(1 - 1/e - \epsilon)$-approximate solution with high probability in the multi-campaign model and show that its expected runtime substantially improves upon that of the algorithm of Budak et al.\ \cite{budak2011limiting}.
\item We give a lower bound of $\Omega(m + n)$ on the time required to obtain a constant approximation for the misinformation mitigation problem.
\item Our experiments show that our algorithm gives an improvement of several orders of magnitude over Budak et al.~\cite{budak2011limiting} and can handle graphs with more than $50$ million edges.
\end{enumerate}    

\section{Related Work}

There exists a large body of work on the Influence Maximization problem first proposed by Kempe et al.\ \cite{kempe2003}. The primary focus of the research community has been related to improving the practical efficiency of the Monte Carlo-based greedy approach under the Independent Cascade (IC) or Linear Threshold (LT) propagation models. These works fall into two categories: heuristics that trade efficiency for approximation guarantees \cite{jung2012irie, wang2012scalable} and practical optimizations that speed up the Monte Carlo-based greedy approach while retaining the approximation guarantees \cite{Leskovec2010, Chen2010, Goyal2013}. Despite these advancements, it remains infeasible to scale the Monte Carlo-based approach to web-scale networks.

Borgs' et al.\ \cite{borgs2012} brought the first asymptotic runtime improvements while maintaining the $(1 - 1/e - \epsilon)$-approximation guarantees with their \emph{reverse influence sampling} technique. Furthermore, they prove their approach is near-optimal under the IC model. Tang et al.\ \cite{Tang2014} presented practical and theoretical improvements to the approach and introduced novel heuristics that result in up to $100$-fold improvements to the runtime.

Incorporating the spread of multiple campaigns is split between two main lines of work: (1) studying influence maximization in the presence of competing campaigns \cite{Bharathi2007, Lin2015, Pathak2010, Li2013} and (2) limiting the spread of misinformation and rumours by launching a truth campaign \cite{budak2011limiting, he2012influence, Fan2013, Nguyen2012containment}. In both cases, existing propagation models (such as IC and LT) are augmented or extended. The work of \cite{budak2011limiting} best captures the idea of preventing the spread of misinformation in a multi-campaign version of the IC model since they aim to minimize the number of users that end up adopting the misinformation. Unfortunately, despite the objective function proving to be monotone and submodular, the Monte Carlo-based greedy solution used in \cite{budak2011limiting} faces the same challenges surrounding scalability as \cite{kempe2003}.

Works \cite{Lin2015, fang2018general} extend the \emph{reverse influence sampling} technique of \cite{borgs2012} to competing campaigns (such as two competing products in \cite{Lin2015} and spreading truth to combat misinformation in \cite{fang2018general}). However, their work differs from ours in an important way: they use a model, different from ours, where the edge probabilities are {\em campaign oblivious}. This alternative model does not capture the notion of misinformation as well as the model we use, but instead is better suited for the influence maximization problem when there are multiple competing campaigns (see \cite{budak2011limiting} for a discussion).

Finally, there exists a large body of work that aims to combat the spread of misinformation from a variety of perspectives beyond those that view the problem through the lens of propagation models such as \cite{shu2017fake, shu2019studying, hassan2017claimbuster, tschiatschek2018fake, kim2018leveraging, pennycook2018crowdsourcing, popat2017truth, jin2016news, wu2018tracing, shiralkar2017finding, simpson2016clearing}.


\section{Preliminaries}

In this section, we formally define the multi-campaign diffusion model, the eventual influence limitation problem presented by Budak et al.\ \cite{budak2011limiting}, and present an overview of the state-of-the-art reverse sampling approach \cite{kempe2003, borgs2012, Tang2014} for the influence maximization problem. 

\vspace{-5pt}

\paragraph*{Diffusion Model} Let $C$ (for ``bad {\em C}ampaign'') and $L$ (for ``{\em L}imiting'') denote two influence campaigns. Let $\mathcal{G} = (V, E, p)$ be a social network with node set $V$ and directed edge set $E$ ($|V| = n$ and $|E| = m$) where $p$ specifies campaign-specific pairwise influence probabilities (or weights) between nodes. That is, $p : E \times Z \rightarrow [0, 1]$ where $Z \in \{ C, L \}$. For convenience, we use $p_Z(e)$ for $p(e, Z)$. Further, let $G = (V, E)$ denote the underlying unweighted directed graph. Given $\mathcal{G}$, the Multi-Campaign Independent Cascade model (MCIC) of Budak et al.\ \cite{budak2011limiting} considers a time-stamped influence propagation process as follows:

\vspace{-5pt}

\begin{enumerate}
\setlength\itemsep{1pt}
\item At timestamp $1$, we \emph{activate} selected sets $A_C$ and $A_L$ of nodes in $\mathcal{G}$ for campaigns $C$ and $L$ respectively, while setting all other nodes \emph{inactive}.
\item If a node $u$ is first activated at timestamp $i$ in campaign $C$ (or $L$), then for each directed edge $e$ that points from $u$ to an inactive neighbour $v$ in $C$ (or $L$), $u$ has $p_C(e)$ (or $p_L(e)$) probability to activate $v$ at timestamp $i + 1$. After timestamp $i + 1$, $u$ cannot activate any node.
\item In the case when two or more nodes from different campaigns are trying to activate $v$ at a given time step we assume that the ``good information" (i.e. campaign $L$) takes effect.
\item Once a node becomes activated in one campaign, it never becomes inactive or changes campaigns.
\end{enumerate}

He et. al.\ \cite{he2012influence} consider the opposite policy to (3) where the misinformation succeeds in the case of a tie-break. We note that our algorithms presented in this work are applicable for both choices of the tie-break policy.

\subsection{Formal Problem Statement}

A natural objective, as outlined in \cite{budak2011limiting}, is ``saving'' as many nodes as possible. That is, we seek to minimize the number of nodes that end up adopting campaign $C$ when the propagation process is complete. This is referred to as the \emph{eventual influence limitation problem (EIL)}.

Let $A_C$ and $A_L$ be the set of nodes from which campaigns $C$ and $L$ start, respectively. Let $I(A_C)$ be the set of nodes that are activated in campaign $C$ in the absence of $L$ when the above propagation process converges and $\pi(A_L)$ be the size of the subset of $I(A_C)$ that campaign $L$ prevents from adopting campaign $C$. We refer to $A_L$ and $A_C$ as the \emph{seed sets}, $I(A_C)$ as the \emph{influence} of campaign $C$, and $\pi(A_L)$ as the \emph{prevention} of campaign $L$. The nodes that are prevented from adopting campaign $C$ are referred to as \emph{saved}. Note that $\pi(A_L)$ is a random variable that depends on the edge probabilities that each node uses in determining out-neighbors to activate.

Budak et al.\ \cite{budak2011limiting} present a simplified version of the problem that captures the idea that it may be much easier to convince a user of the truth. Specifically, the information from campaign $L$ is accepted by users with probability 1 ($p_L(e) = 1$ if edge $e$ exists and $p_L(e) = 0$ otherwise) referred to as the \emph{high effectiveness property}. In \cite{budak2011limiting} it is shown that even with these restrictions EIL with the high effectiveness property is NP-hard. Interestingly, with the high effectiveness property, the prevention function is submodular and thus a Monte Carlo-based greedy approach (referred to here as \emph{MCGreedy}) yields approximation guarantees. 

We motivate the high effectiveness property with the following two real-world scenarios: (1) the phenomenon of ``death hoaxes'' (where celebrities or other notable figures are claimed to have died) have a strong corrective measure when the victim, or a close relative, makes an announcement on their personal account that contradicts the rumour and (2) false reporting of natural disasters can be countered by trusted news organizations providing coverage of the location of the purported scene. In both cases, the sharing of links to strong video, photographic, or text evidence that is also coming from a credible source lends itself to a scenario following the high effectiveness property. In addition to the scenarios we have outlined, the model is attractive because this assumption leads to interesting theoretical guarantees. Budak et al.\ study and obtain results for EIL with the high effectiveness property and is the problem that we consider in this work.



\begin{prob}
Given $\mathcal{G}$, seed set $A_C$, and a positive integer $k$, the eventual influence limitation (EIL) problem asks for a size-$k$ seed set $A_L$ maximizing the value of $\mathbb{E}[\pi(A_L)]$ under the MCIC model with the high effectiveness property.
\end{prob}

\vspace{-10pt}

\paragraph*{Possible Worlds Interpretation} To facilitate a better understanding of MCIC, we define a \emph{Possible World (PW) model} that provides an equivalent view of the MCIC model and follows a widely used convention when studying IM and related problems \cite{kempe2003, budak2011limiting, Chen2010, Goyal2013, Lin2015, he2012influence, budak2011limiting, Fan2013, Nguyen2012containment}. Given a graph $\mathcal{G} = (V, E, p)$ and the MCIC diffusion model, a possible world $X$ consists of two \emph{deterministic graphs}, one for each campaign, sampled from a probability distribution over $\mathcal{G}$. The stochastic diffusion process under the MCIC model has the following equivalent description: we can interpret $\mathcal{G}$ as a distribution over unweighted directed graphs, where each edge $e$ is independently realized with probability $p_C(e)$ (or $p_L(e)$). Observe, given the high effectivness property, the deterministic graph that defines the possible world for campaign $L$ is simply the underlying unweighted graph $G$. Then, if we realize a graph $g$ according to the probability distribution given by $p_C(e)$, we can associate the set of saved nodes in the original process with the set of nodes which campaign $L$ reaches before campaign $C$ during a \emph{deterministic} diffusion process in $g \sim \mathcal{G}$ by campaign $C$ and in $G$ by campaign $L$. That is, we can compute the set of saved nodes with a deterministic cascade in the resulting possible world $X = (g, G)$. The following theorem from \cite{Chen2013} establishes the equivalence between this possible world model and MCIC. This alternative PW model formulation of the EIL problem under the MCIC model will be used throughout the paper.

\begin{thm}[\cite{Chen2013}]
For any fixed seed sets $A_C$ and $A_L$, the joint distributions of the sets of $C$-activated nodes and $L$-activated nodes obtained (i) by running a MCIC diffusion from $A_C$ and $A_L$ and (ii) by randomly sampling a possible world $X = (g, G)$ and running a deterministic cascade from $A_C$ in $g$ and $A_L$ in $G$, are the same.
\end{thm}

\subsection{Reverse Sampling for Influence Maximization}

In this section we review the state-of-the-art approach to the well studied \emph{influence maximization problem (IM)}. This problem is posed in the popular Independent Cascade model (IC) which, unlike the MCIC model, only considers a single campaign. The goal here is to compute a seed set $S_{IM}$ of size $k$ that maximizes the influence of $S_{IM}$ in $\mathcal{G}$. In a small abuse of notation, this section refers to a possible world as the single deterministic graph $g \sim \mathcal{G}$ where each edge in $\mathcal{G}$ is associated with a single influence probability $p(e)$.


Borgs et al.\ \cite{borgs2012} were the first to propose a novel method for solving the IM problem under the IC model that avoids the limitations of the original Monte Carlo-based solution \cite{kempe2003}. Their approach, which was later refined by Tang et al.\ \cite{Tang2014}, is based on the concept of \emph{Reverse Reachable (RR) sets} and is orders of magnitude faster than the greedy algorithm with Monte Carlo simulations, while still providing approximation guarantees with high probability. We follow the convention of \cite{Tang2014} and refer to the method of \cite{borgs2012} as \emph{Reverse Influence Sampling (RIS)}. To explain how \emph{RIS} works, Tang et al.\ \cite{Tang2014} introduce the following definitions:

\begin{defn}[Reverse Reachable Set]
The reverse reachable set for a node $v$ in $g \sim \mathcal{G}$ is the set of nodes that can reach $v$. (That is, for each node $u$ in the RR set, there is a directed path from $u$ to $v$ in $g$.)
\end{defn}

\begin{defn}[Random RR Set]
A random RR set is an RR set generated on an instance of $g \sim \mathcal{G}$, for a node selected uniformly at random from $g$.
\end{defn}

Note, a random RR set encapsulates two levels of randomness: (i) a deterministic graph $g \sim \mathcal{G}$ is sampled where each edge $e \in E$ is independently removed with probability $(1 - p(e))$, and (ii) a ``root'' node $v$ is randomly chosen from $g$. The connection between RR sets and node activation is formalized in the following crucial lemma. 

\begin{lem}
\label{lem:borgs} {\em \cite{borgs2012}}
For any seed set $S$ and node $v$, the probability that an influence propagation process from $S$ can activate $v$ equals the probability that $S$ overlaps an RR set for~$v$.
\end{lem}

\noindent Based on this result, the \emph{RIS} algorithm runs in two steps:

\vspace{-5pt}

\begin{enumerate}
\item Generate random RR sets from $\mathcal{G}$ until a threshold on the total number of steps taken has been reached.
\item Consider the maximum coverage problem of selecting $k$ nodes to cover the maximum number of RR sets generated. Use the standard greedy algorithm for the problem to derive a $(1 - 1/e)$-approximate solution $S_k^{*}$. Return $S_k^{*}$ as the seed set to use for activation.
\end{enumerate}

The rationale behind \emph{RIS} is as follows: if a node $u$ appears in a large number of RR sets it should have a high probability to activate many nodes under the IC model; hence, $u$'s expected influence should be large. As such, we can think of the number of RR sets $u$ appears in as an estimator for $u$'s expected influence. By the same reasoning, if a size-$k$ node set $S_k^{*}$ covers most RR sets, then $S_k^{*}$ is likely to have the maximum expected influence among all size-$k$ node sets in $\mathcal{G}$ leading to a good solution to the IM problem. As shown in \cite{Tang2014}, Lemma \ref{lem:borgs} is the key result that underpins the approximation guarantees of \emph{RIS}. 

The main contribution of Borgs et al.\ is an analysis of their proposed threshold-based approach: \emph{RIS} generates RR sets until the total number of nodes and edges examined during the generation process reaches a pre-defined threshold $\Gamma$. Importantly, $\Gamma$ must be set large enough to ensure a sufficient number of samples have been generated to provide a good estimator for expected influence. They show that when $\Gamma$ is set to $\Theta((m+n)k \log n / \epsilon^2)$, \emph{RIS} runs in near-optimal time $O((m+n)k \log n / \epsilon^2)$, and it returns a $(1 - 1/e - \epsilon)$-approximate solution to the IM problem with at least constant probability.

Due to the more complex dynamics involved in propagation under the MCIC model, adapting the reverse sampling approach to solve EIL is far from trivial.

\section{New Definitions}
\label{sec:def}

\begin{table}
\caption{Frequently used notation.}
\label{tbl:notation}
\centering
\begin{tabular}{ |c|m{0.75\textwidth}| } \hline 
\textbf{Notation}&       \textbf{Description} \\ \hline
$\mathcal{G}$&  a social network represented as a weighted directed graph $\mathcal{G}$ \\ \hline
$G$, $G_T$&     the underlying unweighted graph $G$ and its transpose $G_T$ constructed by reversing the direction of each edge \\ \hline
$g$&            a possible world for campaign $C$ obtained by sampling each edge $e \in \mathcal{G}$ independently with probability $p_C(e)$ \\ \hline
$n$, $m$&       the number of nodes and edges in $\mathcal{G}$ respectively \\ \hline
$k$&            the size of the seed set for misinformation mitigation \\ \hline
$C$, $L$&       the misinformation campaign $C$ and the limiting campaign $L$ \\ \hline
$p_C(e)$, $p_L(e)$& the propagation probability on an edge $e$ for campaigns $C$ and $L$ respectively \\ \hline
$\pi(S)$&       the prevention of a node set $S$ in a misinformation propagation process on $\mathcal{G}$ (see Section \ref{sec:def}) \\ \hline
$\omega(R)$, $\omega_{\pi}(R)$&    the number of edges considered in generating an RRC set and that originate from nodes in an RRC set $R$ (see Equation \ref{eqn:width}) \\ \hline
$\mathcal{R}$&  the set of all RRC sets generated by Algorithm \ref{alg:nodeselection} \\ \hline
$\mathcal{F}_{\mathcal{R}}(S)$&     the fraction of RRC sets in $\mathcal{R}$ that are covered by a node set $S$ \\ \hline
$EPT$&      the expected width of a random RRC set \\ \hline
$OPT_L$&    the maximum $\pi(S)$ for any size-$k$ seed set $S$ \\ \hline
$\lambda$&  see Equation \ref{eqn:lambda} \\ \hline
\end{tabular}
\end{table}

In this section we introduce new definitions that are crucial to the development of our approach. In particular, we formalize the notion of \emph{obstructed} nodes which is required to capture the necessary conditions for saving a node.

\vspace{-5pt}

\paragraph*{Identifying Saved Nodes} Given set $A_L$ of vertices and (unweighted) directed graph $g \sim \mathcal{G}$, write $cl_g(A_L)$ for the set of nodes closer to $A_L$ in $G$ than to $A_C$ in $g$. That is, a node $w \in cl_g(A_L)$ if there exists a node $v$ such that $v \in A_L$ and $|SP_G(v,w)| \le |SP_g(A_C,w)|$ where $SP_H(v,w)$ denotes a shortest path from node $v$ to $w$ in graph $H$ and $SP_H(S,w)$ for a set $S$ denotes the shortest path from any node $v \in S$ to $w$ in graph $H$. When $g$ is drawn from $\mathcal{G}$ this is a necessary, but not sufficient\footnote{In Budak et al.'s work, the set of nodes closer to $A_L$ than $A_C$ is established as a necessary and sufficient condition to \emph{save} a node in the MCIC model, but we note that this should be revised to include our \emph{obstructed} condition due to a gap in the proof of Claim 1 in \cite{budak2011limiting}.}, condition for the set of nodes \emph{saved} by $A_L$. We also require that the nodes in $cl_g(A_L)$ not be \emph{obstructed} by the diffusion of campaign $C$ in $g$.

\begin{defn}[Obstructed Nodes]
\label{defn:blocked}
A node $w \in cl_g(A_L)$ is \emph{obstructed} and cannot be saved by $A_L$ if for every path $p$ from $A_L$ to $w$ there exists a node $u$ on $p$ such that $|SP_g(A_C,u)| < |SP_G(A_L,u)|$.
\end{defn}

Let $obs_g(A_L)$ be the set of obstructed nodes for $A_L$. Conceptually, the nodes in $obs_g(A_L)$ are cutoff because some node on the paths from $A_L$ is reached by campaign $C$ before $L$ which stops the diffusion of $L$.

To help illustrate the concept of obstructed nodes, consider the graph presented in Figure \ref{fig:blocked} and the following possible world instance. Assume that the solid lines are \emph{live} edges that make up the deterministic graph $g \sim \mathcal{G}$ for campaign $C$ in the influence propagation process. The dashed lines are edges that were not realized for campaign $C$. The adversary campaign $C$ starts from $v_c$ while the limiting campaign $L$ starts from $v$. Recall, the deterministic graph $G$ for campaign $L$ in this possible world instance is comprised of \emph{both} the solid and dashed edges due to the high effectiveness property. Observe that $| SP_G(v,w) | = 4$ and $| SP_g(A_C,w) | = 5$. However, $w$ cannot be saved in the resulting cascade since at timestamp 1 the node $u$ will adopt campaign $C$. This intersects the shortest path from $v$ to $w$ and therefore campaign $L$ will not be able to reach node $w$ since a node never switches campaigns. Thus, we say that node $w$ is \emph{obstructed} by $C$.

\begin{figure}
\centering
\begin{tikzpicture}

\tikzset{node/.style={circle,draw,minimum size=0.5cm,inner sep=0pt},}
\tikzset{edge/.style={->,> = latex'},}

\node[node] (v) {$v$};
\node[node] (1) [left = 1cm of v] {};
\node[node] (u) [left = 2.5cm of v] {$u$};
\node[node] (2) [below = 0.5cm of u] {};
\node[node] (w) [below = 0.5cm of 2] {$w$};
\node[node] (v_c) [left = 1cm of u] {$v_c$};
\node[node] (3) [left = 1cm of v_c] {};
\node[node] (4) [below = 0.5cm of 3] {};
\node[node] (5) [below = 0.5cm of 4] {};
\node[node] (6) [left = 1cm of w] {};

\path[draw,thick,->,> = latex']
(v) edge [dashed] node {} (1)
(1) edge [dashed] node {} (u)
(u) edge [dashed] node {} (2)
(2) edge [dashed] node {} (w)
(v_c) edge node {} (u)
(v_c) edge node {} (3)
(3) edge node {} (4)
(4) edge node {} (5)
(5) edge node {} (6)
(6) edge node {} (w);
\end{tikzpicture}
\caption{An example illustrating the concept of obstructed nodes where the possible world graph for campaign $C$ is made up of the solid edges and the possible world for campaign $L$ is made up of both solid and dashed lines.}
\label{fig:blocked}
\vspace{-10pt}
\end{figure}
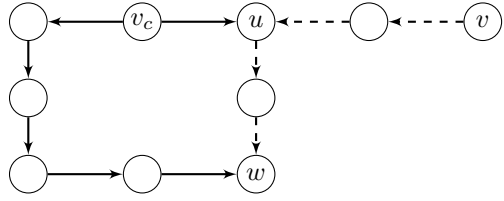

\vspace{-5pt}

\paragraph*{Prevention \& Saviours} Next, we formally define the prevention, $\pi(A_L)$, which corresponds to the number of nodes saved by $A_L$. That is, $\pi(A_L) = | R_g(A_C) \cap ( cl_g(A_L) \setminus obs_g(A_L) ) |$ where $R_H(S)$ is the set of nodes in graph $H$ that are \emph{reachable} from set $S$ (a node $v$ in $H$ is reachable from $S$ if there exists a directed path in $H$ that starts from a node in $S$ and ends at $v$). We write $\mathbb{E}[\pi(A_L)] = \mathbb{E}_{g \sim \mathcal{G}}[\pi(A_L)]$ for the expected prevention of $A_L$ in $\mathcal{G}$. Finally, let $OPT_L = max_{S: |S| = k}\{ \mathbb{E}[\pi(S)] \}$ be the maximum expected prevention of a set of $k$ nodes. 

We refer to the set of nodes that could have saved $u$ as the \emph{saviours} of $u$. A node $w$ is a candidate saviour for $u$ if there is a directed path from $w$ to $u$ in $G$ (i.e.\ reverse reachability). Then, $w$ is a saviour for $u$ subject to the additional constraint that $w$ would not be cutoff by the diffusion of $A_C$ in $g$. That is, a candidate saviour $w$ would be cutoff and cannot be a saviour for $u$ if for every path $p$ from $w$ to $u$ there exists a node $v_b$ such that $|SP_g(A_C,v_b)| < |SP_{G}(w,v_b)|$. We refer to the set of candidate saviours for $u$ that are cutoff as $\tau_g(u)$. Thus, we can define the saviours of $u$ as the set $R_{G^T}(u) \setminus \tau_g(u)$. Therefore, we have:

\begin{defn}[Reverse Reachability without Cutoff Set]
\label{defn:RRC}
The reverse reachability without cutoff (RRC) set for a node $v$ in $g \sim \mathcal{G}$ is the set of saviour nodes of $v$, i.e.\ the set of nodes that can save $v$. (That is, for each node $u$ in the RRC set, $u \in R_{G^T}(u) \setminus \tau_g(u)$.) If $v \not\in R_g(A_C)$ then we define the corresponding RRC set as empty since $v$ is not eligible to be saved.
\end{defn}

\begin{defn}[Random RRC Set]
\label{defn:rRRC}
A random RRC set is an RRC set generated on an instance of $g \sim \mathcal{G}$, for a node selected uniformly at random from $g$.
\end{defn}

\paragraph*{Closing the Gap} Before presenting our reverse sampling approach, we make the following remark regarding obstruction in the context of prior work. The key observation that lead to our definition of obstructed nodes is that the shortest path condition must hold along the \emph{entire} path. This observation was missed by \cite{budak2011limiting} in the MCIC model. Instead, a correct \emph{recursive} definition was provided for the set of nodes that are saved, but the resulting characterization based on shortest paths misses the crucial case of nodes that are obstructed.

Importantly, the solution in \cite{budak2011limiting} can be recovered with a modified proof for Claim 1 and Theorem 4.2. In particular, the statements must include the notion of obstructed nodes in their \emph{inoculation graph} definition, but a careful inspection shows that their objective function remains submodular after this inclusion. As a result, the greedy approach of \cite{budak2011limiting} still provides the stated approximation guarantees and also allows us to incorporate the ideas of \cite{borgs2012} in our solution (as \cite{borgs2012} requires a submodular objective function as well).

\section{Reverse Prevention Sampling}

This section presents our misinformation prevention method, \emph{Reverse Prevention Sampling (RPS)}. At a high level, \emph{RPS}, in the same spirit as \emph{RIS}, consists of two steps. In the first step it derives a parameter $\theta$ that ensures a solution of high quality will be produced. In the second step, using the estimate $\theta$ from step one, it generates $\theta$ RRC sets and then computes the maximum coverage on the resulting collection. More precisely, the two steps are:

\vspace{-5pt}

\begin{enumerate}
\item \textbf{Parameter Estimation. } Compute a lower-bound for the maximum expected prevention among all possible size-$k$ seed sets for $A_L$ and then use the lower-bound to derive a parameter $\theta$.
\item \textbf{Node Selection. } Sample $\theta$ random RRC sets from $\mathcal{G}$ to form a set $\mathcal{R}$ and then compute a size-$k$ seed set $S_k^{*}$ that covers a large number of RRC sets in $\mathcal{R}$. Return $S_k^{*}$ as the final result.
\end{enumerate}

In the rest of this section, we first tackle the challenging task of correctly generating RRC sets in the Node Selection step under the MCIC model. Next, we identify the conditions necessary for the Node Selection of \emph{RPS} to return a solution of good quality and then describe how these conditions are achieved in the Parameter Estimation phase. Table \ref{tbl:notation} provides a reference to some of the frequently used notation. All proofs are in Appendix~A.

\vspace{-5pt}

\paragraph*{Node Selection} The pseudocode of \emph{RPS}'s Node Selection step is presented in Algorithm \ref{alg:nodeselection}. Given $\mathcal{G}$, $k$, $A_C$, and a constant $\theta$ as input, the algorithm stochastically generates $\theta$ random RRC sets, accomplished by repeated invocation of the prevention of misinformation process, and inserts them into a set $\mathcal{R}$. Next, the algorithm follows a greedy approach for the \emph{maximum coverage problem} to select the final seed set. In each iteration, the algorithm selects a node $v_i$ that covers the largest number of RRC sets in $\mathcal{R}$, and then removes all those covered RRC sets from $\mathcal{R}$. The $k$ selected nodes are put into a set $S_k^{*}$, which is returned as the final result.

\begin{algorithm}
\caption{NodeSelection($\mathcal{G}$,$k$,$A_C$,$\theta$)}
\begin{algorithmic}[1]
	\State $\mathcal{R} \gets \emptyset$
	\State Generate $\theta$ random RRC sets and insert them into $\mathcal{R}$.
	\State Initialize a node set $S_k^{*} \gets \emptyset$
	\For{$i$ = 1,\dots,$k$}
		\State Identify the node $v_i$ that covers
		     the most RRC sets in $\mathcal{R}$
		\State Add $v_i$ into $S_k^{*}$
		\State Remove from $\mathcal{R}$ all RRC sets that are covered by $v_i$
	\EndFor \State
	\Return $S_k^{*}$
\end{algorithmic}
\label{alg:nodeselection}
\end{algorithm}

Lines 4-8 in Algorithm \ref{alg:nodeselection} correspond to a standard greedy approach for a \emph{maximum coverage problem}. The problem is equivalent to maximizing a submodular function with cardinality constraints for which it is well known that a greedy approach returns a $(1 - 1/e)$-approximate solution in linear time~\cite{NemhauserVN}.

\subsection{RRC set generation} 

Next, we describe how to generate RRC sets correctly for the EIL problem under the MCIC model, which is more complicated than generating RR sets for the IC model \cite{Tang2014}. The construction of RRC sets is done according to Definition \ref{defn:RRC}. Recall that in the MCIC model, whether a node can be saved or not is based on a number of factors such as whether $v$ is reachable via a path in $g \sim \mathcal{G}$ from $A_C$ and the diffusion history of each campaign. Our algorithms tackle the complex interactions between campaigns by first identifying nodes that can be influenced by $C$ which reveals important information for generating RRC sets for $L$.

Line 2 generates $\mathcal{R}$ by repeated simulation of the misinformation prevention process. The generation of each random RRC set is implemented as two breath-first searches (BFS) on $\mathcal{G}$ and $G^T$ respectively. The first BFS is a \emph{forward labelling} process from $A_C$ implemented as a forward BFS on $\mathcal{G}$ that computes the influence set of $A_C$ in a possible world. The second BFS on $G^T$ is a novel bounded-depth BFS with pruning that carefully tracks which nodes will become obstructed and is described in detail below.

\vspace{-5pt}

\paragraph*{Forward BFS with Lazy Sampling} We first describe the forward labelling process. As the forward labeling is unlikely to reach the whole graph, we simply reveal edge states on demand (``lazy sampling''), based on the principle of deferred decisions. Given the seed set $A_C$ of campaign $C$, we perform a randomized BFS starting from $A_C$ where each outgoing edge $e$ in $\mathcal{G}$ is traversed with $p_C(e)$ probability. The set of nodes traversed in this manner ($R_g(A_C)$) is equivalent to $I(A_C)$ for $g \sim \mathcal{G}$, due to deferred randomness. Note that in each step of the above BFS we record at each node $w$ the minimum distance from $A_C$ to $w$, denoted $D(w)$, for use in the second BFS.


Given a randomly selected node $u$ in $G$, observe that for $u$ to be able to be saved we require $u \in R_g(A_C)$. Therefore, if the randomly selected node $u \not\in R_g(A_C)$ then we return an empty RRC set. On the other hand, if $u \in R_g(A_C)$, we have $D(u) = |SP_g(A_C,u)|$ as a result of the above randomized BFS which indicates the maximum distance from $u$ that candidate saviour nodes can exist. We run a second BFS from $u$ in $G^T$ to depth $D(u)$ to determine the saviour nodes for $u$ by carefully pruning those nodes that would become obstructed.

\vspace{-5pt}

\paragraph*{Bounded-depth BFS with Pruning} The second BFS on $G^T$, presented in Algorithm \ref{alg:genRRC}, takes as input a source node $u$, the maximum depth $D(u)$, and a directed graph $G^T$. Algorithm \ref{alg:genRRC} utilizes special indicator values associated with each node $w$ to account for potential cutoffs from $C$. Each node $w$ holds a variable, $\beta(w)$, which indicates the distance beyond $w$ that the BFS can go before the diffusion would have been cutoff by $C$ propagating in $g$. The $\beta$ value for each node $w$ is initialized to $D(w)$. In each round, the current node $w$ has an opportunity to update the $\beta$ value of each of its successors only if $\beta(w) > 0$. For each successor $z$ of $w$, we assign $\beta(z) = \beta(w) - 1$ if $\beta(z) = \texttt{null}$ or if $\beta(z) > 0$ and $\beta(w) - 1 < \beta(z)$. In this way, each ancestor of $z$ will have an opportunity to apply a $\beta$ value to $z$ to ensure that if any ancestor has a $\beta$ value then so will $z$ and furthermore, the $\beta$ variable for $z$ will be updated with the smallest $\beta$ value from its ancestors. We terminate the BFS early if we reach a node $w$ with $\beta(w) = 0$.

Figure \ref{fig:genRRC} captures the primary scenarios encountered by Algorithm \ref{alg:genRRC} when initialized at $u$. The enclosing dotted line represents the extent of the influence of campaign $C$ for the current influence propagation process. First, notice that if the BFS moves away from $A_C = \{ v_c \}$, as in the case of node $z$, that, once we move beyond the influence boundary of $C$, there will be no potential for cutoff. As such, the BFS is free to traverse until the maximum depth $D(u)$ is reached. On the other hand, if the BFS moves towards (or perpendicular to) $v_c$ then we must carefully account for potential cutoff. For example, when the BFS reaches $v$, we know the distance from $v_c$ to $v$: $D(v) = SP_g(v_c,v)$. Therefore, the BFS must track the fact that there cannot exist saviours at a distance $D(v)$ beyond $v$. In other words, if we imagine initializing a misinformation prevention process from a node $w$ such that $SP_G(v,w) > D(v)$ then $v$ will adopt campaign $C$ before campaign $L$ can reach $v$. Therefore, at each out-neighbour of $v$ we use the knowledge of $D(v)$ to track the distance beyond $v$ that saviours can exist. This updating process tracks the smallest such value and is allowed to cross the enclosing influence boundary of campaign $C$ ensuring that all potential for cutoff is tracked.

\begin{figure}
\centering
\def\svgwidth{0.6\textwidth}
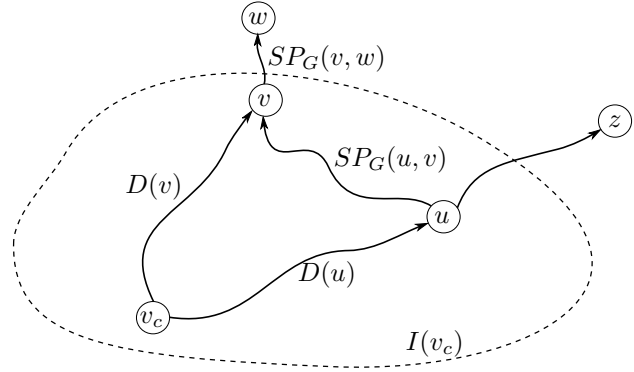
\caption{An overview of the primary scenarios encountered by Algorithm \ref{alg:genRRC}.}
\label{fig:genRRC}
\vspace{-10pt}
\end{figure}

Finally, we collect all nodes visited during the process (including $u$), and use them to form an RRC set. The runtime of this procedure is precisely the sum of the degrees (in $G$) of the nodes in $R_g(A_C)$ plus the sum of the degrees of the nodes in $R_{G^T}(u) \setminus \tau(u)$.

We briefly note another key difference between \emph{RPS} and \emph{RIS} occurs in the RRC set generation step. Unlike in the single campaign setting, generating an RRC set is comprised of two phases instead of just one. First, we are required to simulate the spread of misinformation since being influenced by campaign $C$ is a pre-condition for being saved. As a result, only a fraction of the simulation steps of \emph{RPS} provide signal for the prevention value we are trying to estimate. This difference is made concrete in the running time analysis to follow.

\begin{algorithm}
\caption{generateRRC($u$, $D(u)$, $G^T$)}
\begin{algorithmic}[1]
	\State let $R \gets \emptyset$, $Q$ be a queue and $Q.enqueue(u)$
	\State set $u.depth = 0$ and label $u$ as discovered
	\While{$Q$ is not empty}
		\State $w \gets Q.dequeue()$, $R \gets R \cup \{ w \}$
		\If{$w.depth = D(u)$ OR $\beta(w) = 0$}
			\State \textbf{continue}
		\EndIf
		\ForAll{nodes $z$ in $G^T.adjacentEdges(w)$}
			\If{$\beta(w) > 0$ AND $\beta(z) > 0$}
				\If{$\beta(w) - 1 < \beta(z)$}
					\State $\beta(z) \gets \beta(w) - 1$
				\EndIf
			\ElsIf{$\beta(w) > 0$}
				\State $\beta(z) \gets \beta(w) - 1$
			\EndIf
			\If{$z$ is not labelled as discovered}
				\State set $z.depth = w.depth + 1$, label $z$ as discovered and $Q.enqueue(z)$
			\EndIf
		\EndFor
	\EndWhile
	\State \Return $R$
\end{algorithmic}
\label{alg:genRRC}
\end{algorithm}

\vspace{-10pt}

\subsection{Analysis}

In this section we focus on two parameters: solution quality and runtime. For Algorithm \ref{alg:nodeselection} to return a solution with approximation guarantee, we will provide a lower bound on $\theta$. Then, we will analyze the running time of the algorithm in terms of $\theta$ and a quantity $EPT$ that captures the expected number of edges traversed when generating a random RRC set.

\vspace{-5pt}

\paragraph*{Approximation Guarantee} We begin by establishing the crucial connection between RRC sets and the prevention process on $\mathcal{G}$. That is, the prevention of a set of nodes $S$ is precisely $n$ times the probability that a node $u$, chosen uniformly at random, has a saviour from $S$.

\begin{lem}
\label{lem:equiv}
For any seed set $S$ and any node $v$, the probability that a prevention process from $S$ can save $v$ equals the probability that $S$ overlaps an RRC set for $v$.
\end{lem}

\begin{proof}
Let $S$ be a fixed set of nodes, and $v$ be a fixed node. Suppose that we generate an RRC set $R$ for $v$ on a graph $g \sim \mathcal{G}$. Let $\rho_1$ be the probability that $S$ overlaps with $R$ and let $\rho_2$ be the probability that $S$, when used as a seed set, can save $v$ in a prevention process on $\mathcal{G}$. By Definition \ref{defn:RRC}, if $v \in R_g(A_C)$ then $\rho_1$ equals the probability that a node $u \in S$ is a saviour for $v$. That is, $\rho_1$ equals the probability that $G$ contains a directed path from $u \in S$ to $v$ and $u \not\in \tau(v)$ and $0$ if $v \not\in R_g(A_C)$. Meanwhile, if $v \in R_g(A_C)$ then $\rho_2$ equals the probability that a node $u \in S$ can save $v$ (i.e. $v \in ( cl_g(u) \setminus obs_g(u) )$) and $0$ if $v \not\in R_g(A_C)$. It follows that $\rho_1 = \rho_2$ due to the equivalence between the set of saviours for $v$ and the ability to save $v$.
\end{proof}

For any node set $S$, let $F_{\mathcal{R}}(S)$ be the fraction of RRC sets in $\mathcal{R}$ covered by $S$. Then, based on Lemma \ref{lem:equiv}, we can prove that the expected value of $n \cdot F_{\mathcal{R}}(S)$ equals the expected prevention of $S$ in $\mathcal{G}$.

\begin{cor}
\label{cor:equiv}
$\mathbb{E}[n \cdot F_{\mathcal{R}}(S)] = \mathbb{E}[\pi(S)]$
\end{cor}

\begin{proof}
Observe that $\mathbb{E}[F_{\mathcal{R}}(S)]$ equals the probability that $S$ intersects a random RRC set, while $\mathbb{E}[\pi(S)]/n$ equals the probability that a randomly selected node can be saved by $S$ in a prevention process on $\mathcal{G}$. By Lemma \ref{lem:equiv}, the two probabilities are equal, leading to $\mathbb{E}[n \cdot F_{\mathcal{R}}(S)] = \mathbb{E}[\pi(S)]$.
\end{proof}

Corollary \ref{cor:equiv} implies that we can estimate $\mathbb{E}[\pi(S)]$ by estimating the fraction of RRC sets in $\mathcal{R}$ covered by $S$. The number of sets covered by a node $v$ in $\mathcal{R}$ is precisely the number of times we observed that $v$ was a saviour for a randomly selected node $u$. We can therefore think of $n \cdot F_{\mathcal{R}}(S)$ as an estimator for $\mathbb{E}[\pi(S)]$. Our primary task is to show that it is a \emph{good} estimator. Using Chernoff bounds, we show that $n \cdot F_{\mathcal{R}}(S)$ is an accurate estimator of any node set $S$'s expected prevention, when $\theta$ is sufficiently large:

\begin{lem}
\label{lem:bound}
Suppose that $\theta$ satisfies

\vspace{-10pt}

\begin{equation}
\label{eqn:theta}
\theta \ge (8 + 2 \epsilon) n \cdot \frac{l \log n + \log \binom{n}{k} + \log 2}{OPT_L \cdot \epsilon^2}
\end{equation}

Then, for any set $S$ of at most $k$ nodes, the following inequality holds with at least $1 - n^{-l} / \binom{n}{k}$ probability:

\vspace{-5pt}

\begin{equation}
\label{eqn:estimator}
\Big | n \cdot F_{\mathcal{R}}(S) - \mathbb{E}[{\pi(S)}] \Big | < \frac{\epsilon}{2} \cdot OPT_L
\end{equation}
\end{lem}

\begin{proof}
Let $\rho$ be the probability that $S$ overlaps with a random RRC set. Then, $\theta \cdot F_{\mathcal{R}}(S)$ can be regarded as the sum of $\theta$ i.i.d. Bernoulli variables with a mean $\rho$. By Corollary \ref{cor:equiv},

\begin{equation*}
\rho = \mathbb{E}[F_{\mathcal{R}}(S)] = \mathbb{E}[\pi(S)] / n
\end{equation*}

\noindent Then, we have

\begin{gather}
Pr \bigg [ | n \cdot F_{\mathcal{R}}(S) - \mathbb{E}[{\pi(S)}] | \ge \frac{\epsilon}{2} \cdot OPT_L \bigg ] \nonumber \\
\hspace{5em} = Pr \bigg [ | \theta \cdot F_{\mathcal{R}}(S) - \rho \theta | \ge \frac{\epsilon \theta}{2 n} \cdot OPT_L \bigg ] \nonumber \\
\hspace{6em} = Pr \bigg [ | \theta \cdot F_{\mathcal{R}}(S) - \rho \theta | \ge \frac{\epsilon \cdot OPT_L}{2 n \rho} \cdot \rho \theta \bigg ] \label{eqn:estimator_proof}
\end{gather}

Let $\delta = \epsilon \cdot OPT_L / (2 n \rho)$. By the Chernoff bounds, Equation \ref{eqn:theta}, and the fact that $\rho = \mathbb{E}[\pi(S)] / n \le OPT_L / n$, we have

\begin{align*}
\text{r.h.s. of Eqn.\ \ref{eqn:estimator_proof} } & < 2 exp \Big ( - \frac{\delta^2}{2 + \delta} \cdot \rho \theta \Big ) \\
& = 2 exp \Big ( - \frac{\epsilon^2 \cdot OPT_L^2}{8 n^2 \rho + 2 \epsilon n \cdot OPT_L} \cdot \theta \Big ) \\
& \le 2 exp \Big ( - \frac{\epsilon^2 \cdot OPT_L^2}{8 n \cdot OPT_L + 2 \epsilon n \cdot OPT_L} \cdot \theta \Big ) \\
& = 2 exp \Big ( - \frac{\epsilon^2 \cdot OPT_L}{(8 + 2 \epsilon) \cdot n} \cdot \theta \Big ) \le \frac{1}{\binom{n}{k} \cdot n^l}
\end{align*} 

Therefore, the lemma is proved.
\end{proof}

Based on Lemma \ref{lem:bound}, we prove that if Eqn.\ \ref{eqn:theta} holds, Algorithm \ref{alg:nodeselection} returns a $(1 - 1/e - \epsilon)$-approximate solution with high probability by a simple application of Chernoff bounds.

\begin{thm}
\label{thm:main}
Given a $\theta$ that satisfies Equation \ref{eqn:theta}, Algorithm \ref{alg:nodeselection} returns a $(1 - 1/e - \epsilon)$-approximate solution with at least $1 - n^{-l}$ probability.
\end{thm}

\begin{proof}
Let $S_k$ be the node set returned by Algorithm \ref{alg:nodeselection}, and $S_k^+$ be the size-$k$ node set that maximizes $F_{\mathcal{R}}(S_k^+)$ (i.e., $S_k^+$ covers the largest number of RRC sets in $\mathcal{R}$). As $S_k$ is derived from $\mathcal{R}$ using a $(1-1/e)$-approximate algorithm for the maximum coverage problem, we have $F_{\mathcal{R}}(S_k) \ge (1-1/e) \cdot F_{\mathcal{R}}(S_k^+)$. Let $S_k^{\circ}$ be the optimal solution for the EIL problem on $\mathcal{G}$, i.e. $\mathbb{E}[\pi(S_k^{\circ})] = OPT_L$. We have $F_{\mathcal{R}}(S_k^+) \ge F_{\mathcal{R}}(S_k^{\circ})$, which leads to $F_{\mathcal{R}}(S_k) \ge (1-1/e) \cdot F_{\mathcal{R}}(S_k^{\circ})$.

Assume that $\theta$ satisfies Equation \ref{eqn:theta}. By Lemma \ref{lem:bound}, Equation \ref{eqn:estimator} holds with at least $1 - n^{-l} / \binom{n}{k}$ probability for any given size-$k$ node set $S$. Thus, by the union bound, Equation \ref{eqn:estimator} should hold simultaneously for all size-$k$ node sets with at least $1 - n^{-l}$ probability. In that case, we have

\begin{align*}
\mathbb{E}[\pi(S_k)] & > n \cdot F_{\mathcal{R}}(S_k) - \epsilon / 2 \cdot OPT_L \\
				& \ge (1-1/e) \cdot n \cdot F_{\mathcal{R}}(S_k^+) - \epsilon / 2 \cdot OPT_L \\
				& \ge (1-1/e) \cdot n \cdot F_{\mathcal{R}}(S_k^{\circ}) - \epsilon / 2 \cdot OPT_L \\
				& \ge (1-1/e) \cdot (1 - \epsilon / 2) \cdot OPT_L - \epsilon / 2 \cdot OPT_L \\
				& > (1-1/e - \epsilon) \cdot OPT_L
\end{align*}

Thus, the theorem is proved.
\end{proof}

\paragraph*{Runtime} First, we will define $EPT$ which captures the expected number of edges traversed when generating a random RRC set. After that, we define the expected runtime of \emph{RPS} in terms of $EPT$ and the parameter $\theta$.

Let $M_R$ be the instance of $R_g(A_C)$ used in computing an RRC set $R$. Then, we define the \emph{width} of an RRC set $R$, denoted as $\omega(R)$, as the number of edges in $G$ that point to nodes in $R$ plus the number of edges in $G$ that originate from nodes in $M_R$. That is

\vspace{-10pt}

\begin{equation}
\label{eqn:width}
\omega(R) = \sum_{u \in M_R} outdegree_G(u) + \sum_{v \in R} indegree_G(v)
\end{equation}


Let $EPT$ be the expected width of a random RRC set, where the expectation is taken over the randomness in $R$ and $M_R$, and observe that Algorithm \ref{alg:nodeselection} has an expected runtime of $O(\theta \cdot EPT)$. This can be observed by noting that $EPT$ captures the expected number of edge traversals required to generate a random RRC set since an edge is only considered in the propagation process (either of the two BFS's) if it points to a node in $R$ or originates from a node in $M_R$. An important consideration is that, since $OPT_L$ is unknown, we cannot set $\theta$ directly from Equation \ref{eqn:theta}. For simplicity, we define

\vspace{-10pt}

\begin{equation}
\label{eqn:lambda}
\lambda = (8 + 2 \epsilon) n \cdot \Big (l \log n + \log \binom{n}{k} + \log 2 \Big ) \cdot \epsilon^{-2}
\end{equation}

and rewrite Equation \ref{eqn:theta} as $\theta \ge \lambda / OPT_L$. In the parameter estimation step we employ the techniques of \cite{Tang2014} to derive a $\theta$ value for \emph{RPS} that is above the threshold but also allows for practical efficiency.

\subsection{Parameter Estimation}
\label{sec:param_est}

Our objective in this section is to identify a $\theta$ that makes $\theta \cdot EPT$ reasonably small, while still ensuring $\theta \ge \lambda / OPT_L$. We begin with some definitions. Let $\mathcal{V^{*}}$ be a probability distribution over the nodes in $G$, such that the probability mass for each node is proportional to its indegree in $G$. Let $v^{*}$ be a random variable following $\mathcal{V^{*}}$ and recall that $M_R$ is a random instance of $R_g(A_C)$ that is equivalent to the influence $I(A_C)$ for a possible world $g$. Furthermore, define $\omega(M_R)$, the number of edges in $G$ that originate from nodes in $M_R$, as $\omega(M_R) = \sum_{u \in M_R} outdegree_G(u)$. Then we prove the following.

\begin{lem}
\label{lem:EPT}
$\frac{m}{n} \cdot \mathbb{E}[\pi(\{ v^{*} \})] = EPT - \mathbb{E}[\omega(M_R)]$, where the expectation of $\pi(\{ v^{*} \})$ and $\omega(M_R)$ is taken over the randomness in $v^{*}$ and the prevention process.
\end{lem}

\begin{proof}

Let $R$ be a random RRC set, $M_R$ be the random instance of $R_g(A_C)$ used to compute $R$, and $p_R$ be the probability that a randomly selected edge from $G$ points to a node in $R$. Then, $EPT = \mathbb{E}[\omega(M_R)] + \mathbb{E}[p_R \cdot m]$, where the expectation is taken over the random choices of $R$ and $M_R$.

Let $v^{*}$ be a sample from $\mathcal{V^{*}}$ and $b(v^{*}, R)$ be a boolean function that returns $1$ if $v^{*} \in R$, and $0$ otherwise. Then, for any fixed $R$, $p_R = \sum_{v^{*}} (\Pr[v^{*}] \cdot b(v^{*},R))$. Now consider that we fix $v^{*}$ and vary $R$. Define $p_{v^{*}, R} = \sum_{R} (\Pr[R] \cdot b(v^{*},R))$ so that by Lemma \ref{lem:equiv}, $p_{v^{*}, R}$ equals the probability that a randomly selected node can be saved in a prevention process when $\{ v^{*} \}$ is used as a seed set. Therefore, $\mathbb{E}[p_{v^{*}, R}] = \mathbb{E}[\pi(\{ v^{*} \})] / n$. This leads to

\vspace{-10pt}

\begin{align*}
\mathbb{E}[p_R] &= \sum_{R} (\Pr[R] \cdot p_R) = \sum_R (\Pr[R] \cdot \sum_{v^{*}} (\Pr[v^{*}] \cdot b(v^{*}, R) ) ) \\
				&= \sum_{v^{*}} (\Pr[v^{*}] \cdot \sum_R (\Pr[R] \cdot b(v^{*}, R) ) ) = \sum_{v^{*}} (\Pr[v^{*}] \cdot p_{v^{*}, R} ) = \mathbb{E}[p_{v^{*}, R}] =  \mathbb{E}[\pi(\{ v^{*} \})] / n
\end{align*}

Therefore, $EPT = \mathbb{E}[\omega(M_R)] + m \cdot \mathbb{E}[p_R] = \mathbb{E}[\omega(M_R)] + \frac{m}{n} \cdot \mathbb{E}[\pi(\{ v^{*} \})]$. This completes the proof.
\end{proof}

Lemma \ref{lem:EPT} shows that if we randomly sample a node from $\mathcal{V^{*}}$ and calculate its expected prevention $p$, then on average we have $p = \frac{n}{m} ( EPT - \mathbb{E}[\omega(M_R)] )$. This implies that $\frac{n}{m} ( EPT - \mathbb{E}[\omega(M_R)] ) \le OPT_L$, since $OPT_L$ is the maximum expected prevention of any size-$k$ node set.

Recall that the expected runtime complexity of Algorithm \ref{alg:nodeselection} is $O(\theta \cdot EPT)$. Now, suppose we are able to identify a parameter $t$ such that $t = \Omega(\frac{n}{m} ( EPT - \mathbb{E}[\omega(M_R)] ))$ and $t \le OPT_L$. Then, by setting $\theta = \lambda / t$, we can guarantee that Algorithm \ref{alg:nodeselection} is correct, since $\theta \ge \lambda / OPT_L$, and has an expected runtime complexity of

\vspace{-10pt}

\begin{equation}
\label{eqn:runtime}
O(\theta \cdot EPT) = O \Big (\frac{\lambda}{t} \cdot EPT \Big ) = O \Bigg (\frac{\lambda \cdot EPT}{\frac{n}{m} ( EPT - \mathbb{E}[\omega(M_R)] )} \Bigg )
\end{equation}

Furthermore, if we define a ratio $\gamma \in (0,1)$ which captures the relationship between $\mathbb{E}[\omega(M_R)]$ and $EPT$ by writing $\mathbb{E}[\omega(M_R)] = \gamma EPT$, we can rewrite Equation \ref{eqn:runtime} as

\vspace{-10pt}

\begin{equation}
\label{eqn:runtime_proper}
O \Big (\frac{m}{n} \Big (\frac{1}{1 - \gamma} \Big ) \lambda \Big ) = O((k + l)(m + n)(1 / (1 - \gamma)) \log n / \epsilon^2 )
\end{equation}

Note that $\gamma$ is a data-dependent approximation factor not present in \cite{Tang2014}, but arises from the MCIC model. In particular, the RRC set generation relies crucially on first computing the spread of misinformation from campaign $C$ in order to determine the set of nodes that can be saved. See Section \ref{sec:lowerbound} for a detailed discussion of $\gamma$.

\vspace{-5pt}


\paragraph*{Computing $t$} Ideally, we seek a $t$ that increases monotonically with $k$ to mimic the behaviour of $OPT_L$. Suppose we take $k$ samples $e_i = (v_i,w_i)$ with replacement over a uniform distribution on the edges in $G$, and use the $v_i$'s to form a node set $S^{*}$. Let $KPT$ be the mean of the expected prevention of $S^{*}$ over the randomness in $S^{*}$ and the prevention process. Due to the submodularity of the prevention function, it can be verified that $KPT$ increases with $k$ and

\vspace{-5pt}

\begin{equation}
\label{eqn:ept_kpt}
\frac{n}{m} \Big ( EPT - \mathbb{E}[\omega(M_R)] \Big ) \le KPT \le OPT_L
\end{equation}

\vspace{5pt}

Additionally,

\begin{lem}
\label{lem:kpt}
Let $R$ be a random RRC set and define the \emph{subwidth} of $R$ ($\omega_{\pi}(R)$), the number of edges in $G$ that point to nodes in $R$, as $\omega_{\pi}(R) = \sum_{v \in R} indegree_G(v)$. Define

\vspace{-5pt}

\begin{equation}
\label{eqn:kappa}
\kappa(R) = 1 - \bigg (1 - \frac{\omega_{\pi}(R)}{m} \bigg )^k
\end{equation}

Then, $KPT = n \cdot \mathbb{E}[\kappa(R)]$, where the expectation is taken over the random choices of $R$.
\end{lem}

\vspace{-5pt}

\begin{proof}
Let $S^{*}$ be a node set formed by the $v_i$ from $k$ samples $e_i = (v_i,w_i)$ over a uniform distribution on the edges in $G$, with duplicates removed. Let $R$ be a random RRC set, and $\alpha_R$ be the probability that $S^{*}$ overlaps with $R$. Then, by Corollary \ref{cor:equiv},

\vspace{-5pt}

\begin{equation*}
KPT = \mathbb{E}[\pi(S^{*})] = \mathbb{E}[n \cdot \alpha_R]
\end{equation*}

Consider that we sample $k$ times over a uniform distribution on the edges in $G$. Let $E^{*}$ be the set of edges sampled, with duplicates removed. Let $\alpha'_R$ be the probability that one of the edges in $E^{*}$ points to a node in $R$. It can be verified that $\alpha'_R = \alpha_R$. Furthermore, given that there are $\omega_{\pi}(R)$ edges in $G$ that point to nodes in $R$, $\alpha'_R = 1 - (1 - \omega_{\pi}(R)/m)^k = \kappa(R)$. Therefore,

\vspace{-5pt}

\begin{equation*}
KPT  = \mathbb{E}[n \cdot \alpha_R] = \mathbb{E}[n \cdot \alpha'_R] = \mathbb{E}[n \cdot \kappa(R)]
\end{equation*}

Which proves the lemma.
\end{proof}

Lemma \ref{lem:kpt} shows we can estimate $KPT$ by computing $n \cdot \kappa(R)$ on a set of random RRC sets and averaging over a sufficiently large sample size. However, if we want to obtain an estimate of $KPT$ with $\delta \in (0,1)$ relative error with at least $1 - n^{-l}$ probability then Chernoff bounds dictate that the number of samples required is $\Omega(l n \log n \cdot \delta^{-2} / KPT)$ which depends on $KPT$ itself. This issue is also encountered in \cite{Tang2014} and we are able to resolve it by mimicking their adaptive sampling approach, which dynamically adjusts the number of measurements based on the observed values of $\kappa(R)$, under the MCIC model. Importantly, great care is required in order to confirm that the same proof techniques used in \cite{Tang2014} apply to our setting. At a high level, our proofs follow the same logic, but differ in a several subtle ways due to the added complexity of the MCIC model. We will highlight these differences in the analysis to follow.

\begin{algorithm}
\caption{KptEstimation($\mathcal{G}$,$k$,$A_C$)}
\begin{algorithmic}[1]
	\For{$i = 1$ to $2 \log_2n - 1$}
		\State Let $c_i = (6 l \log n + 6 \log (\log_2 n)) \cdot 2^i$
		\State Let $sum = 0$
		\For{$j = 1$ to $c_i$}
			\State Generate a random RRC set $R$
			\State $\kappa(R) = 1 - (1 - \frac{\omega_{\pi}(R)}{m})^k$
			\State $sum = sum + \kappa(R)$
		\EndFor
		\If{$sum / c_i > 1 / 2^i$} \State
			\Return $KPT^{*} = n \cdot sum / (2 \cdot c_i)$
		\EndIf
	\EndFor \State
	\Return $KPT^{*} = 1 / n$
\end{algorithmic}
\label{alg:kptestimation}
\end{algorithm}

\vspace{5pt}

\paragraph*{Estimating \emph{KPT}} The sampling approach for estimating $KPT$ is presented in Algorithm \ref{alg:kptestimation}. We begin with a high level description of how the algorithm proceeds. Since $KPT$ is an unknown quantity, we begin with the assumption that it takes on the value $n/2$. Then, we compute an estimate for the expected value of $\kappa(R)$ based on a relatively few number of samples. Chernoff bounds allow us to determine if the computed value of $KPT = n \cdot \kappa(R)$ is a good estimator and, if so, the algorithm terminates. However, if the estimate is much smaller than $n/2$ we apply the standard doubling approach and generate an increased number of samples to determine if $KPT$ takes on a value close to half the initial estimate. The algorithm continues computing estimates for $KPT$, based on an increasing number of samples, and comparing to values that halve in size until the error bounds dictated by Chernoff bounds indicate we have reached a suitably precise estimation of $KPT$.

In each iteration (Lines 2-7), the goal of Algorithm \ref{alg:kptestimation} is to compute the average value of $\kappa(R)$ over $c_i$ randomly generated RRC sets from $\mathcal{G}$. As described in Lemma's \ref{lem:six} and \ref{lem:seven} below, the $c_i$ are chosen carefully such that if the average computed for $\kappa(R)$ over the $c_i$ samples is greater than $2^{-i}$ then we can conclude that we have a good estimate for $KPT$ with high probability. More precisely, that the expected value of $\kappa(R)$ is at least half of the average computed.  Conversely, if the average computed is too small then Chernoff bounds imply that we have a bad estimate for $KPT$ and the algorithm proceeds to the next iteration.

The IC problem benefits from a lower bound on $KPT$ of $1$ which allows the algorithm to terminate in $\log_2 n - 1$ iterations. However, we do not benefit from this lower bound for the MCIC problem, since a seed node for campaign $C$ is not guaranteed to save any nodes in $M_R$. In the case that the true value of $KPT$ is very small, the algorithm will terminate in the $2 \log_2 n$-th iteration and return $KPT^{*} = 1 / n$, which equals the smallest possible $KPT$ assuming $| M_R | \ge 2$ always holds. This lower bound follows directly from the probability that a randomly selected seed used for computing $KPT$ falls within $M_R$. As we will show in the next section, $KPT^{*} \in [KPT/4,OPT_L]$ holds with a high probability and therefore setting $\theta = \lambda / KPT^{*}$ ensures Algorithm \ref{alg:nodeselection} is correct and achieves the expected runtime complexity in Equation \ref{eqn:runtime_proper}.

\vspace{5pt}

\paragraph*{Performance Bounds} Proving the correctness and demonstrating bounds on the runtime for Algorithm \ref{alg:kptestimation} requires a careful analysis of the algorithm's behaviour. As shown below, we make use of two lemmas to prove that the algorithm's estimate of $KPT^{*}$ is close to $KPT$.

Let $\mathcal{K}$ be the distribution of $\kappa(R)$ over random RRC sets in $\mathcal{G}$ with domain $[0,1]$. Let $\mu = KPT / n$, and $s_i$ be the sum of $c_i$ i.i.d. samples from $\mathcal{K}$, where $c_i$ is defined as $c_i = (6 l \log n + 6 \log (\log_2 n)) \cdot 2^i$. Chernoff bounds give

\begin{lem}
\label{lem:six}
If $\mu \le 2^{-j}$, then for any $i \in [1,j-1]$,

\vspace{-5pt}

\begin{equation}
Pr \Bigg [ \frac{s_i}{c_i} > \frac{1}{2^i} \Bigg ] < \frac{1}{n^l \cdot \log_2 n}
\end{equation}
\end{lem}

\vspace{-5pt}

\begin{proof}
Let $\delta = (2^{-i} - \mu)/\mu$. By the Chernoff bounds,

\begin{align*}
Pr \Bigg [ \frac{s_i}{c_i} > 2^{-i} \Bigg ] &\le exp \Bigg ( - \frac{\delta^2}{2 + \delta} \cdot c_i \cdot \mu \Bigg ) \\
		&= exp( - c_i \cdot (2^{-i} - \mu)^2 / (2^{-i} + \mu) ) \\
		&\le exp( - c_i \cdot 2^{-i-1} / 3 ) = \frac{1}{n^l \cdot \log_2 n}
\end{align*}

This completes the proof.
\end{proof}

By Lemma \ref{lem:six}, if $KPT \le 2^{-j}$, then Algorithm \ref{alg:kptestimation} is very unlikely to terminate in any of the first $j - 1$ iterations. This prevents the algorithm from outputting a $KPT^{*}$ too much larger than $KPT$.

\begin{lem}
\label{lem:seven}
If $\mu \ge 2^{-j}$, then for any $i \ge j + 1$,

\vspace{-5pt}

\begin{equation}
Pr \Bigg [ \frac{s_i}{c_i} > \frac{1}{2^i} \Bigg ] > 1 - \Big ( \frac{1}{n^l \cdot \log_2 n} \Big )^{2^{i-j-1}}
\end{equation}
\end{lem}

\vspace{-5pt}

\begin{proof}
Let $\delta = (\mu - 2^{-i})/\mu$. By the Chernoff bounds,

\begin{align*}
Pr \Bigg [ \frac{s_i}{c_i} \le 2^{-i} \Bigg ] &\le exp \Bigg ( - \frac{\delta^2}{2} \cdot c_i \cdot \mu \Bigg ) \\
		&= exp( - c_i \cdot (\mu - 2^{-i})^2 / (2 \cdot \mu) ) \\
		&\le exp( - c_i \cdot \mu / 8) < \Big ( \frac{1}{n^l \cdot \log_2 n} \Big )^{2^{i-j-1}}
\end{align*}

This completes the proof.
\end{proof}

By Lemma \ref{lem:seven}, if $KPT \le 2^{-j}$ and Algorithm \ref{alg:kptestimation} enters iteration $i > j+1$, then it will terminate in the $i$-th iteration with high probability. This ensures that the algorithm does not output a $KPT^{*}$ that is too much smaller than $KPT$.

Based on Lemmas \ref{lem:six} and \ref{lem:seven}, we prove the following theorem for the correctness and expected runtime of Algorithm \ref{alg:kptestimation}.

\begin{thm}
\label{thm:kpt}
When $n \ge 2$ and $l \ge 1/2$, Algorithm \ref{alg:kptestimation} returns $KPT^{*} \in [KPT/4,OPT_L]$ with at least $1 - 2 n^{-l}$ probability, and runs in $O(l(m+n)(1 / (1 - \gamma)) \log n)$ expected time. Furthermore, $\mathbb{E}[\frac{1}{KPT^{*}}] < \frac{12}{KPT}$.
\end{thm}

\vspace{-5pt}

\begin{proof}
Assume that $KPT / n \in [2^{-j}, 2^{-j+1}]$. We first prove the accuracy of the $KPT^{*}$ returned by Algorithm \ref{alg:kptestimation}.

By Lemma \ref{lem:six} and the union bound, Algorithm \ref{alg:kptestimation} terminates in or before the $(j-2)$-th iteration with less than $n^{-l} (j-2) / \log_2 n$ probability. On the other hand, if Algorithm \ref{alg:kptestimation} reaches the $(j+1)$-th iteration, then by Lemma \ref{lem:seven}, it terminates in the $(j+1)$-th iteration with at least $1 - n^{-l} / \log_2 n$ probability. Given the union bound and the fact that Algorithm \ref{alg:kptestimation} has at most $2 \log_2 n - 1$ iterations, Algorithm \ref{alg:kptestimation} should terminate in the $(j-1)$-th, $j$-th, or $(j+1)$-th iteration with a probability at least $1 - n^{-l} (2 \log_2 n - 2) / \log_2 n$. In that case, $KPT^{*}$ must be larger than $n/2 \cdot 2^{-j-1}$, which leads to $KPT^{*} > KPT / 4$. Furthermore, $KPT^{*}$ should be $n/2$ times the average of at least $c_{j-1}$ i.i.d. samples from $\mathcal{K}$. By the Chernoff bounds, it can be verified that

\vspace{-5pt}

\begin{equation*}
Pr[ KPT^{*} \ge KPT ] \le n^{-l} / \log_2 n
\end{equation*}

\vspace{5pt}

By the union bound, Algorithm \ref{alg:kptestimation} returns, with probability at least $1 - 2 n^{-l}$ probability, $KPT^{*} \in [ KPT/4, KPT ] \subseteq [ KPT/4, OPT_L ]$.

Next, we analyze the expected runtime of Algorithm \ref{alg:kptestimation}. Recall that the $i$-th iteration of the algorithm generates $c_i$ RRC sets, and each RRC set takes $O(EPT)$ expected time. Given that $c_{i+1} = 2 \cdot c_i$ for any $i$, the first $j+1$ iterations generate less than $2 \cdot c_{j+1}$ RRC sets in total. Meanwhile, for any $i' \ge j+2$, Lemma \ref{lem:seven} shows that Algorithm \ref{alg:kptestimation} has at most $n^{-l \cdot 2^{i' - j - 1}} / \log_2 n$ probability to reach the $i'$-th iteration. Therefore, when $n \ge 2$ and $l \ge 1/2$, the expected number of RRC sets generated after the first $j+1$ iterations is less than

\vspace{-5pt}

\begin{equation*}
\sum_{i' = j+2}^{2 \log_2 n - 1} \Big ( c_{i'} \cdot \Big ( \frac{1}{n^l \cdot \log_2 n} \Big )^{2^{i-j-1}} \Big ) < c_{j+2}
\end{equation*}

\vspace{5pt}

Hence, the expected total number of RRC sets generated by Algorithm \ref{alg:kptestimation} is less than $2 c_{j+1} + c_{j+2} = 2 c_{j+2}$. Therefore, the expected time complexity of the algorithm is

\begin{align*}
O(c_{j+2} \cdot EPT) &= O(2^j l \log n \cdot EPT) \\
	&= O \Big (2^j l \log n \cdot \Big (1 + \frac{m}{n} \Big ) \cdot KPT \cdot \Big ( \frac{1}{1 - \gamma} \Big ) \Big ) \\
	&= O \Big (2^j l \log n \cdot (m + n) \cdot 2^{-j} \cdot \Big ( \frac{1}{1 - \gamma} \Big ) \Big ) \\
	&= O \Big (l \log n \cdot (m + n) \cdot \Big ( \frac{1}{1 - \gamma} \Big ) \Big ) \\
\end{align*}

\vspace{-5pt}

Here we used Equation \ref{eqn:ept_kpt} in the second equality. Finally, we show that $\mathbb{E}[1/KPT^{*}] < 12/KPT$. Observe that if Algorithm \ref{alg:kptestimation} terminates in the $i$-th iteration, it returns $KPT^{*} \ge n \cdot 2^{-i-1}$. Let $\zeta_i$ denote the event that Algorithm \ref{alg:kptestimation} stops in the $i$-th iteration. By Lemma \ref{lem:seven}, when $n \ge 2$ and $l \ge 1/2$, we have

\vspace{-5pt}

\begin{align*}
\mathbb{E}[1/KPT^{*}] &= \sum_{i=1}^{2 \log_2 n - 1} \Big ( 2^{i+1} / n \cdot \Pr[\zeta_i] \Big ) \\
	& < \sum_{i=j+2}^{2 \log_2 n - 1} \Big ( 2^{i+1} / n \cdot \Big ( n^{-l \cdot 2^{i-j-1}} / \log_2 n \Big ) \Big ) + 2^{j+2} / n \\
	& < (2^{j+3} + 2^{j+2}) / n \le 12 / KPT
\end{align*}

This completes the proof.
\end{proof}




\subsection{Improved Parameter Estimation}

This section describes a new heuristic for improving the practical performance of \emph{RPS}, without affecting its asymptotic guarantees, by improving the estimated lower bound on $OPT_L$. Our heuristic simplifies the one introduced in \cite{Tang2014} and further, is adapted to the MCIC setting.

\begin{algorithm}
\caption{RefineKPT($\mathcal{G}$,$k$,$A_C$, $KPT^{*}$, $\epsilon'$)}
\begin{algorithmic}[1]
	\State Let $\lambda' = (2 + \epsilon') l n \log n \cdot (\epsilon')^{-2}$.
	\State Let $\theta' = \lambda' / KPT^{*}$.
	\State Generate $\theta'$ random RRC sets; put them into a set $\mathcal{R}'$.
	\State Initialize a node set $S'_k = \emptyset$.
	\For{$i = 1$ to $k$}
		\State Identify node $v_i$ that covers the most RRC sets in~$\mathcal{R}'$.
		\State Add $v_i$ to $S'_k$.
		\State Remove from $\mathcal{R}'$ all RRC sets that are covered by $v_i$.
	\EndFor
	\State Let $f$ be the fraction of the original $\theta'$ RRC sets that are covered by $S'_k$.
	\State Let $KPT' = f \cdot n / (1 + \epsilon')$
	\State \Return $KPT^{+} = max\{KPT', KPT^{*}\}$
\end{algorithmic}
\label{alg:refinekpt}
\end{algorithm}

Observe that the $KPT^{*}$ output by Algorithm \ref{alg:kptestimation} largely determines the efficiency of \emph{RPS}. If $KPT^{*}$ is close to $OPT_L$, then $\theta = \lambda / KPT^{*}$ is small and Algorithm \ref{alg:nodeselection} only needs to generate a relatively small number of RRC sets. However, if $KPT^{*} \ll OPT_L$ then the efficiency of Algorithm \ref{alg:nodeselection} degrades rapidly and, in turn, so does the overall performance of \emph{RPS}.

To remedy this issue, we can add an intermediate step before Algorithm \ref{alg:nodeselection} to refine $KPT^{*}$ into a potentially tighter lower-bound of $OPT_L$. The intuition behind this heuristic is to generate a reduced number $\theta'$ of random RRC sets, placing them into a set $\mathcal{R}'$, and then apply the greedy approach (for the maximum coverage problem) on $\mathcal{R}'$ to obtain a good estimator for the maximum expected prevention in $\mathcal{R}'$. Thus, we can use the estimation as a good lower-bound for $OPT_L$.

Algorithm \ref{alg:refinekpt} shows the pseudo-code of the intermediate step. It first generates $\theta'$ random RRC sets and invokes the greedy approach for the maximum coverage problem on $\mathcal{R}'$ to obtain a size-$k$ node set $S'_{k}$. Algorithm \ref{alg:refinekpt} computes the fraction $f$ of RRC sets that are covered by $S'_{k}$ so that, by Corollary \ref{cor:equiv}, $f \cdot n$ is an unbiased estimation of $\mathbb{E}[\pi(S'_{k})]$. We set $\theta'$ based on the $KPT^{*}$ output by Algorithm \ref{alg:kptestimation} to a reasonably large number to ensure that $f \cdot n < (1 + \epsilon') \cdot \mathbb{E}[\pi(S'_{k})]$ occurs with at least $1 - n^{-l}$ probability. Based on this, Algorithm \ref{alg:refinekpt} computes $KPT' = f \cdot n / (1 + \epsilon')$ scaling $f \cdot n$ down by a factor of $1 + \epsilon'$ to ensure that $KPT' \le \mathbb{E}[\pi(S'_{k})] \le OPT_L$. The final output of Algorithm \ref{alg:refinekpt} is $KPT^{+} = max\{KPT', KPT^{*}\}$. Below we give a lemma that shows the theoretical guarantees of Algorithm \ref{alg:refinekpt}.

\begin{lem}
Given that $\mathbb{E}[\frac{1}{KPT^{*}}] < \frac{12}{KPT}$, Algorithm \ref{alg:refinekpt} runs in $O(l (m+n) (1 / (1 - \gamma)) \log n / (\epsilon')^2)$ expected time. In addition, it returns $KPT^{+} \in [KPT^{*}, OPT_L]$ with at least $1 - n^{-l}$ probability, if $KPT^{*} \in [KPT/4, OPT_L]$.
\end{lem}

\vspace{-5pt}

\begin{proof}
We first analyze the expected time complexity of Algorithm \ref{alg:refinekpt}. Observe that the expected time complexity of Lines 1-3 of Algorithm \ref{alg:refinekpt} is $O(\mathbb{E}[\frac{\lambda'}{KPT^{*}}] \cdot EPT)$, since they generate $\frac{\lambda'}{KPT^{*}}$ random RRC sets, each of which takes $O(EPT)$ expected time to generate. By Theorem \ref{thm:kpt}, $\mathbb{E}[\frac{1}{KPT^{*}}] < \frac{12}{KPT}$. 

\vspace{-5pt}

\begin{align*}
O \bigg ( \mathbb{E} \bigg [ \frac{\lambda'}{KPT^{*}} \bigg ] \cdot EPT \bigg )
    &= O \bigg ( \frac{\lambda'}{KPT} \cdot EPT \bigg ) \\
	&= O \bigg ( \frac{\lambda'}{KPT} \cdot \bigg ( 1 + \frac{m}{n} \bigg ) \cdot KPT \cdot \bigg ( \frac{1}{1 - \gamma} \bigg ) \bigg ) \\
	&= O \bigg ( \lambda' \cdot \bigg ( 1 + \frac{m}{n} \bigg ) \cdot  \bigg ( \frac{1}{1 - \gamma} \bigg ) \bigg ) \\
	&= O(l (m + n) (1 / (1 - \gamma)) \log n / (\epsilon')^2)
\end{align*}

\noindent On the other hand, Lines 4-12 run in time linear to the total size of the RRC sets in $\mathcal{R}'$, i.e.\ the set of all RRC sets generated in Lines 1-3 of Algorithm \ref{alg:refinekpt}. Given that the expected total size of the RRC sets in $\mathcal{R}'$ should be no more than $O(l (m+n) (1 / (1 - \gamma)) \log n)$, Lines 4-12 of Algorithm \ref{alg:refinekpt} have an expected time complexity of $O(l (m+n) (1 / (1 - \gamma)) \log n)$. Therefore, the expected time complexity of Algorithm \ref{alg:refinekpt} is $O(l (m + n) (1 / (1 - \gamma)) \log n / (\epsilon')^2)$.

Next, we prove that Algorithm \ref{alg:refinekpt} returns $KPT^{+} \in [KPT^{*}, OPT_L]$ with high probability. First, observe that $KPT^{+} \ge KPT^{*}$ trivially holds, as Algorithm \ref{alg:refinekpt} sets $KPT^{+} = max \{ KPT', KPT^{*} \}$, where $KPT'$ is derived in Line 11 of Algorithm \ref{alg:refinekpt}. To show that $KPT^{+} \in [KPT^{*}, OPT_L]$, it suffices to prove that $KPT' \le OPT_L$.

By Line 11 of Algorithm \ref{alg:refinekpt}, $KPT' = f \cdot n / (1 + \epsilon')$, where $f$ is the fraction of RRC sets in $\mathcal{R}'$ that is covered by $S'_k$, where $\mathcal{R}'$ is a set of $\theta'$ random RRC sets, and $S'_k$ is a size-$k$ node set generated from Lines 4-9 in Algorithm \ref{alg:refinekpt}. Therefore, $KPT' \le OPT_L$ if and only if $f \cdot n \le (1 + \epsilon') \cdot OPT_L$.

Let $\rho'$ be the probability that a random RRC set is covered by $S'_k$. By Corollary \ref{cor:equiv}, $\rho' = \mathbb{E}[\pi(S'_k)] / n$. In addition, $f \cdot \theta'$ can be regarded as the sum of $\theta'$ i.i.d.\ Bernoulii variables with a mean $\rho'$. Therefore, we have

\vspace{-5pt}

\begin{align}
\Pr[f \cdot n & > (1 + \epsilon') \cdot OPT_L] \nonumber \\
    &\le \Pr \big [ n \cdot f - \mathbb{E}[\pi(S'_k)] > \epsilon' \cdot OPT_L \big ] \nonumber \\
	& = \Pr \bigg [ \theta' \cdot f - \theta' \cdot \rho' > \frac{\theta'}{n} \cdot \epsilon' \cdot OPT_L \bigg ] \nonumber \\
	& = \Pr \bigg [ \theta' \cdot f - \theta' \cdot \rho' > \frac{\epsilon' \cdot OPT_L}{n \cdot \rho'} \cdot \theta' \cdot \rho' \bigg ] \label{eqn:refineone}
\end{align}

\noindent Let $\delta = \epsilon' \cdot OPT_L / (n \rho')$. By the Chernoff bounds, we have

\vspace{-5pt}

\begin{align*}
\text{r.h.s.\ of Eqn. \ref{eqn:refineone}} & \le \exp \bigg ( - \frac{\delta^2}{2 + \delta} \cdot \rho' \theta' \bigg ) \\
	& = \exp \bigg ( - \frac{\epsilon'^2 \cdot OPT_{L}^2}{2 n^2 \rho' + \epsilon' n \cdot OPT_L} \cdot \theta' \bigg ) \\
	& \le \exp \bigg ( - \frac{\epsilon'^2 \cdot OPT_{L}^2}{2 n \cdot OPT_L + \epsilon' n \cdot OPT_L} \cdot \theta' \bigg ) \\
	& = \exp \bigg ( - \frac{\epsilon'^2 \cdot OPT_{L}}{(2 + \epsilon') \cdot n} \cdot \frac{\lambda'}{KPT^{*}} \bigg ) \\
	& \le \exp \bigg ( - \frac{\epsilon'^2 \cdot \lambda'}{(2 + \epsilon') \cdot n} \bigg ) \le \frac{1}{n^l} \\
\end{align*}

\vspace{-5pt}

\noindent Therefore, $KPT' = f \cdot n / (1 + \epsilon') \le OPT_L$ holds with at least $1 - n^{-l}$ probability. This completes the proof.
\end{proof}

Note that the time complexity of Algorithm \ref{alg:refinekpt} is smaller than that of Algorithm \ref{alg:kptestimation} by a factor of $k$, since the former only needs to accurately estimate the expected prevention of one node set (i.e. $S'_{k}$), whereas the latter needs to ensure accurate estimations for $\binom{n}{k}$ node sets simultaneously. Additionally, our new approach eliminates the need to first compute a seed set from the RRC sets generated in the last iteration of Algorithm \ref{alg:kptestimation} as in \cite{Tang2014}. 

\paragraph*{Wrapping Up} In summary, we integrate Algorithm \ref{alg:refinekpt} into \emph{RPS} and obtain an improved solution (referred to as \emph{$RPS^{+}$}) as follows. Given $\mathcal{G}$, $k$, $A_C$, $\epsilon$, and $l$, we first invoke Algorithm \ref{alg:kptestimation} to derive $KPT^{*}$. After that, we feed $\mathcal{G}$, $k$, $A_C$, $KPT^{*}$, and a parameter $\epsilon^{*}$ (as defined in \cite{Tang2014}) to Algorithm \ref{alg:refinekpt}, and obtain $KPT^{+}$ in return. Then, we compute $\theta = \lambda / KPT^{+}$, where $\lambda$ is as defined in Equation \ref{eqn:lambda}. Finally, we run Algorithm \ref{alg:nodeselection} with $\mathcal{G}$, $k$, $A_C$, and $\theta$ as the input and get the output $S_k^{*}$ as the final result of prevention maximization. 

By Theorems \ref{thm:main} and \ref{thm:kpt}, Equation \ref{eqn:runtime_proper}, and the union bound, \emph{RPS} runs in $O((k+l)(m+n)(1 / (1 - \gamma)) \log n / \epsilon^2)$ expected time and it can be verified that when $\epsilon' \ge \epsilon / \sqrt{k}$, \emph{$RPS^{+}$} has the same time complexity as \emph{RPS}. Furthermore, \emph{$RPS^{+}$} returns a $(1 - 1/e - \epsilon)$-approximate solution with at least $1 - 4n^{-l}$ probability and the success probability can be increased to $1 - n^{-l}$ by scaling $l$ up by a factor of $1 + \log4 / \log n$.

Finally, we note that the time complexity of \emph{RPS} is \emph{near-optimal} up to the instance-specific factor $\gamma$ under the MCIC model, as it is only a $(\frac{1}{1 - \gamma}) \log n$ factor larger than the $\Omega(m+n)$ lower-bound proved in Section \ref{sec:lowerbound} (for fixed $k$, $l$, and $\epsilon$).

\section{Lower Bounds}
\label{sec:lowerbound}

\vspace{-5pt}

\paragraph*{Comparison with \emph{MCGreedy}} \emph{MCGreedy} runs in $O(kmnr)$ time, where $r$ is the number of Monte Carlo samples used to estimate the expected prevention of each node set. Budak et al. do not provide a detailed analysis related to how $r$ should be set to achieve a $(1 - 1/e - \epsilon)$-approximation ratio in the MCIC model, only pointing out that when each estimation of expected prevention has $\epsilon$ relative error, \emph{MCGreedy} returns a $(1 - 1/e - \epsilon')$-approximate solution for a particular $\epsilon'$ \cite{budak2011limiting}. In the following lemma, we present a more precise characterization of the relationship between $r$ and \emph{MCGreedy}'s approximation ratio in the MCIC model.

\begin{lem}
\emph{MCGreedy} returns a $(1 - 1/e - \epsilon)$-approximate solution with at least $1 - n^{-l}$ probability, if

\vspace{-10pt}

\begin{equation}
\label{eqn:greedy}
r \ge (8k^2 + 2k \epsilon) \cdot n \cdot \frac{(l+1) \log n + \log k}{\epsilon^2 \cdot OPT_L}
\end{equation}
\label{lem:greedy_comp}
\end{lem}

\begin{proof}
Let $S$ be any node set that contains no more than $k$ nodes in $G$, and $\xi(S)$ be an estimation of $\mathbb{E}[\pi(S)]$ using $r$ Monte Carlo steps. We first prove that, if $r$ satisfies Equation \ref{eqn:greedy}, then $\xi(S)$ will be close to $\mathbb{E}[\pi(S)]$ with a high probability.

Let $\mu = \mathbb{E}[\pi(S)] / n$ and $\delta = \epsilon OPT_L / (2 k n \mu)$. By the Chernoff bounds, we have

\begin{align}
\Pr \big [ & | \xi(S) - \mathbb{E}[\pi(S)] | > \frac{\epsilon}{2k} OPT_L \big ] \nonumber \\
    &= \Pr \bigg [ \bigg | r \cdot \frac{\xi(S)}{n} - r \cdot \frac{\mathbb{E}[\pi(S)]}{n} \bigg | > \frac{\epsilon}{2 k n} \cdot r \cdot OPT_L \bigg ] \nonumber \\
	&= \Pr \bigg [ \bigg | r \cdot \frac{\xi(S)}{n} - r \cdot \frac{\mathbb{E}[\pi(S)]}{n} \bigg | > \delta \cdot r \cdot \mu \bigg ] \nonumber \\
	&< 2 \exp \bigg ( - \frac{\delta^2}{2 + \delta} \cdot r \cdot \mu \bigg ) \nonumber \\
	&= 2 \exp \bigg ( - \frac{\epsilon^2}{(8 k^2 + 2 k \epsilon) \cdot n} \cdot r \cdot \mu \bigg ) \nonumber \\
	&= 2 \exp ((l + 1) \log n + \log k)  \nonumber\\
	&= \frac{1}{k \cdot n^{l + 1}} \label{eqn:greedyprob}
\end{align}

Observe that, given $\mathcal{G}$, $k$, and $A_C$ \emph{MCGreedy} runs in $k$ iterations, each of which estimates the expected prevention f at most $n$ node sets with sizes no more than $k$. Therefore, the total number of node sets inspected by \emph{MCGreedy} is at most $k n$. By Equation \ref{eqn:greedyprob} and the union bound, with at least $1 - n^{-l}$ probability, we have

\begin{equation}
\label{eqn:greedyone}
| \xi(S') - \mathbb{E}[\pi(S')] | \le \frac{\epsilon}{2k} OPT_L
\end{equation}

\noindent for all those $k n$ node sets $S'$ simultaneously. In what follows, we analyze the accuracy of \emph{MCGreedy}'s output, under the assumption that for any node set $S'$ considered by \emph{MCGreedy}, it obtains a sample of $\xi(S')$ that satisfies Equation \ref{eqn:greedyprob}. For convenience, we abuse notation and use $\xi(S')$ to denote the aforementioned sample.

Let $S_0 = \emptyset$, and $S_i \quad (i \in [1,k])$ be the node set selected by \emph{MCGreedy} in the $i$-th iteration. We define $x_i = OPT_L - \pi(S_i)$, and $y_i(v) = \pi(S_{i-1} \cup \{ v \} ) - \pi(s_{i-1}))$ for any node $v$. Let $v_i$ be the node that maximizes $y_i(v_i)$. Then, $y_i(v_i) \ge x_{i-1} / k$ must hold; otherwise, for any size-$k$ node set $S$, we have

\begin{align*}
\pi(S) & \le \pi(S_{i-1}) + \pi(S \setminus S_{i-1}) \\
	& \le \pi(S_{i-1}) + k \cdot y_i(v_i) \\
	& < \pi(S_{i-1} + x_{i-1}) = OPT_L
\end{align*}

\noindent which contradicts the definition of $OPT_L$.

Recall that, in each iteration of \emph{MCGreedy}, it add into $S_{i-1}$ the node $v$ that leads to the largest $\xi(S_{i-1} \cup \{ v \})$. Therefore,

\begin{equation}
\label{eqn:greedytwo}
\xi(S_i) - \xi(S_{i-1}) \ge \xi(S_i \cup \{ v \}) - \xi(S_{i-1})
\end{equation}

Combining Equations \ref{eqn:greedyone} and \ref{eqn:greedytwo}, we have

\begin{align}
x_{i-1} & - x_i = \pi(S_i) - \pi(S_{i-1}) \nonumber \\
	& \ge \xi(S_i) - \frac{\epsilon}{2 k} OPT_L - \xi(S_{i-1}) + \big ( \xi(S_{i-1}) - \pi(S_{i-1}) \big ) \nonumber \\
	& \ge \xi(S_{i-1} \cup \{ v_i \}) - \xi(S_{i-1}) - \frac{\epsilon}{2 k} OPT_L + \big ( \xi(S_{i-1}) - \pi(S_{i-1}) \big ) \nonumber \\
	& \ge \pi \big ( S_{i-1} \cup \{ v_i \} \big ) - \pi(S_{i-1}) - \frac{\epsilon}{k} OPT_L \nonumber \\
	& \ge \frac{1}{k} x_{i-1} - \frac{\epsilon}{k} OPT_L \label{eqn:greedythree}
\end{align}

\noindent Equation \ref{eqn:greedythree} leads to

\begin{align*}
x_k & \le \big ( 1 - \frac{1}{k} \big ) \cdot x_{k-1} + \frac{\epsilon}{k} OPT_L \nonumber \\
	& \le \big ( 1 - \frac{1}{k} \big )^2 \cdot x_{k-2} + \big ( 1 + \big ( 1 - \frac{1}{k} \big ) \big ) \cdot \frac{\epsilon}{k} OPT_L \\
	& \le \big ( 1 - \frac{1}{k} \big )^k \cdot x_0 + \sum_{i = 0}^{k - 1} \big ( \big ( 1 - \frac{1}{k} \big )^i \cdot \frac{\epsilon}{k} OPT_L \big ) \\
	& = \big ( 1 - \frac{1}{k} \big )^k \cdot OPT_L + \big ( 1 - \big ( 1 - \frac{1}{k} \big )^k \big ) \cdot \epsilon \cdot OPT_L \\
	& \le \frac{1}{e} \cdot OPT_L - \big ( 1 - \frac{1}{e} \big ) \cdot \epsilon \cdot OPT_L
\end{align*}

\noindent Therefore,

\begin{align*}
\pi(S_k) & = OPT_L - x_k \\
	& \le (1 - 1/e) \cdot (1 - \epsilon) \cdot OPT_L \\
	& \le (1 - 1/e - \epsilon) \cdot OPT_L
\end{align*}

\noindent Thus, the lemma is proved.
\end{proof}

\vspace{-10pt}

Assume that we know $OPT_L$ in advance and set $r$ to the smallest value satisfying the above inequality, in \emph{MCGreedy}'s favour. In that case, the time complexity of \emph{MCGreedy} is $O(k^3 l m n^2 \epsilon^{-2} \log n / OPT_L)$. Towards comparing \emph{MCGreedy} to \emph{RPS}, we show the following upper bound on the value of $\gamma$.

\begin{clm}
\label{clm:gamma}
$\gamma \le \frac{n}{n+1}$
\end{clm}

\begin{proof}
Let $M_R$ be the random instance of $R_g(A_C)$ used to compute $R$. Furthermore, let us assume that $|M_R| \ge 2$ so that at least one non-seed node is influenced by campaign $C$. Then, from Lemma \ref{lem:EPT} and the definition of $\gamma$ we have

\begin{equation*}
\frac{1}{\gamma} = 1 + \frac{m}{n} \cdot \frac{\mathbb{E}[\pi(\{ v^* \})]}{\mathbb{E}[\omega(M_R)]}
\end{equation*}

Then, observe that the expected number of nodes saved by $v^*$ is at least $\Pr[v^* \in M_R]$. That is, if $v^* \in M_R$ then campaign $L$ can save at least one node, namely $v^*$ itself. Giving

\vspace{-10pt}

\begin{align*}
\Pr[v^* \in M_R] & = \sum_{v \in M_R} \frac{deg(v)}{m} \ge \sum_{v \in M_R} \frac{1}{m} = \frac{|M_R|}{m}
\end{align*}

Therefore, $\frac{\mathbb{E}[\pi(\{ v^* \})]}{\mathbb{E}[\omega(M_R)]} \ge \frac{1}{m}$. Then, we have $\frac{1}{\gamma} \ge 1 + \frac{m}{n} \cdot \frac{1}{m} = 1 + \frac{1}{n}$. Thus, we get $\gamma \le \frac{n}{n+1}$ proving the claim.
\end{proof}

Claim \ref{clm:gamma} shows that the expected runtime for \emph{RPS} is at most $O((k+l) m n \epsilon^{-2} \log n)$. As a consequence, given that $OPT_L \le n$, the expected runtime of \emph{MCGreedy} is always more than the expected runtime of \emph{RPS}. In practice, we observe that for typical social networks $OPT_L \ll n$ and $\frac{1}{1 - \gamma} \ll n+1$ resulting in superior scalability of \emph{RPS} compared to \emph{MCGreedy}. Table~\ref{tbl:gamma} confirms that $\frac{1}{1 - \gamma} \ll n+1$ on our small datasets.

\vspace{-5pt}

\paragraph*{A Lower Bound for EIL} In the theorem below, we provide a lower bound on the time it takes for any algorithm to compute a $\beta$-approximation for the EIL problem given uniform node sampling and an adjacency list representation. Thus, we rule out the possibility of a sublinear time algorithm for the EIL problem for an arbitrary $\beta$.

\begin{thm}
\label{thm:lowerbound}
Let $0 < \epsilon < \frac{1}{10 e}$, $\beta \le 1$ be given. Any randomized algorithm for EIL that returns a set of seed nodes with approximation ratio $\beta$, with probability at least $1 - \frac{1}{e} - \epsilon$, must have a runtime of at least $\frac{\beta (m + n)}{24 \min \{ k, 1/\beta \}}$.
\end{thm}

\begin{proof}
We make use of Yao's Minmax Lemma for the performance of Las Vegas (LV) randomized algorithms on a family of inputs \cite{yao1977}. Precisely, the lemma states that the least expected cost of deterministic LV algorithms on a distribution over a family of inputs lower bounds the expected cost of the optimal randomized LV algorithm over that family of inputs. We build such an input family of lower bound graphs via the use of a novel gadget.

Throughout the proof we assume all edge probabilities for both campaigns are $1$. Note, for a graph consisting of $p = n/2$ connected pairs for which each pair contains a node $u \in A_C$, an algorithm must return at least $\beta k$ nodes to obtain an approximation ratio of $\beta$. Doing so in at most $\beta^2 n / 2$ queries requires that $2 \beta k \le \beta^2 n$, which implies $2 k / \beta \le n$. We can therefore assume $2 k / \beta \le n$.

Define the cost of the algorithm as $0$ if it returns a set of seed nodes with approximation ratio better than $\beta$ and $1$ otherwise. As the cost of an algorithm equals its probability of failure, we can think of it as a LV algorithm. Assume for notational simplicity that $\beta = 1 / T$ where $T$ is an integer. We will build a family of lower bound graphs, one for each value of $n$ (beginning from $n = 1 + T$); each graph will have $m \le n$, so it will suffice to demonstrate a lower bound of $\frac{n}{12 T \min \{ k, T \}}$.

We now consider the behaviour of a deterministic algorithm $A$ with respect to the uniform distribution on the constructed family of inputs. For a given value $T$ the graph would be made from $k$ components of size $2 T$ and $p = \frac{n - 2 k T}{2}$ connected pairs (recall that $2 k T = 2 k / \beta \le n$). Specifically, the $k$ components of size $2 T$ are structured as follows: for each component there is a \emph{hub} node $v_h$ that is connected to $2 T - 2$ leaf nodes and a node $u \in A_C$. Furthermore, each of the $p$ pairs also contains one node $u \in A_C$. If algorithm $A$ returns seed nodes from $l$ of the $k$ components of size $2 T$, it achieves a total prevention of $l \cdot (2 T - 1) + (k - l)$ since choosing either the hub node $v_h$ or any leaf node will result in saving all $2 T - 1$ eligible nodes in the component. Thus, to attain an approximation factor better than $\beta = \frac{1}{T}$, we must have $l \cdot (2 T - 1) + (k - l) \ge \frac{1}{T} \cdot k \cdot (2 T - 1)$, which implies $l \ge \frac{k}{2 T}$ for any $T > 1$.

Suppose $k > 12 T$. The condition $l \ge \frac{k}{2 T}$ implies that at least $\frac{k}{2 T}$ of the large components must be queried by the algorithm, where each random query has probability $\frac{k \cdot (2 T - 1)}{n-(p+k)} \ge \frac{k T}{n}$ of hitting a large component. If the algorithm makes fewer than $\frac{n}{6 T^2}$ queries, then the expected number of components hit is $\frac{n}{6 T^2} \cdot \frac{k T}{n} = \frac{k}{6 T}$. Chernoff bounds then imply that the probability of hitting more than $\frac{k}{2 T}$ components is no more than $e^{- \frac{k}{6 T} \cdot 2 / 3} \le \frac{1}{e^{4 / 3}} < 1 - \frac{1}{e} - \epsilon$, a contradiction.

If $k \le 12 T$ then we need that $l \ge 1$, which occurs only if the algorithm queries at least one of the $k \cdot (2 T - 1)$ vertices in the large components. With $\frac{n}{k \cdot (2 T - 1)}$ queries, for $n$ large enough, this happens with probability less than $\frac{1}{e} - \epsilon$, a contradiction.

We conclude that, in all cases, at least $\frac{n}{ 24 T \min \{ k, T \}}$ queries are necessary to obtain an approximation factor better than $\beta = \frac{1}{T}$ with probability at least $1 - \frac{1}{e} - \epsilon$, as required. By Yao's Minmax Principle this gives a lower bound of $\Omega(\frac{n}{24 T \min \{ k, T \}})$ on the expected performance of any randomized algorithm, on at least one of the inputs.

Finally, the construction can be modified to apply to non-sparse networks. For any $d \le n$, we can augment our graph by overlaying a $d$-regular graph with exponentially small weight on each edge. This does not significantly impact the prevention of any set, but increases the time to decide if a node is in a large component by a factor of $O(d)$ (as edges must be traversed until one with non-exponentially-small weight is found). Thus, for each $d \le n$, we have a lower bound of $\Omega(\frac{n d}{24 T \min \{ k, T \}})$ on the expected performance of $A$ on a distribution of networks with $m = n d$ edges.
\end{proof}

\section{Generalization to the Multi-Campaign Triggering Model}

The \emph{triggering model} is an influence propagation model that generalizes the IC and LT models. It assumes that each node $v$ is associated with a triggering distribution $\mathcal{T}(v)$ over the power set of $v$'s incoming neighbors. An influence propagation process under the triggering model works as follows: (1) for each node $v$, take a sample from $\mathcal{T}(v)$ and define the sample as the triggering set of $v$, then (2) at timestep 1 activate the seed set $S$, and (3) in subsequent timesteps, if an active node appears in the triggering set of $v$, then $v$ becomes active. The propagation terminates when no more nodes can be activated.

We can define a \emph{multi-campaign} version of the triggering model (MCT) that generalizes the MCIC model by associating each node with a \emph{campaign-specific} triggering distribution $\mathcal{T}_Z(v)$ where $Z \in \{ C, L \}$. The propagation process under MCT proceeds exactly as under MCIC with the exception that activation between rounds (step 2) is determined by $\mathcal{T}_C(v)$ and $\mathcal{T}_L(v)$. To the best of our knowledge, we are the first to formally define a multi-campaign version of the triggering model.

The key aspect of the MCIC model that enabled the existence of obstructed nodes is that the two campaigns are allowed to propagate along \emph{different sets of edges} in a possible world $X$. This is exactly the intuition captured by the example in Figure \ref{fig:blocked} and is caused by $L$ and $C$ having separate propagation probabilities in $\mathcal{G}$. As a result, the campaigns traverse potentially unique graphs in $X$ and results in the possibility of the obstruction of $L$ by $C$. This observation holds under the more general setting of MCT due to the campaign-specific triggering sets and so the obstruction phenomenon exists under the MCT model.

\begin{rem}
Propagation under the multi-campaign triggering model requires the notion of obstruction to correctly characterize the conditions required to save an arbitrary node.
\end{rem}

Following the observations made in \cite{Tang2014}, our solutions can be easily extended to operate under the multi-campaign triggering (MCT) model with a modified high effectiveness property. Under MCT, the high effectiveness property asserts that $\mathcal{T}_L(v) = in(v)$ where $in(v)$ is the set of in-neighbours of $v$ in $G$. Observe that Algorithm \ref{alg:nodeselection} does not rely on anything specific to the MCIC model, except a subroutine to generate random RRC sets. Thus, we can revise the definition of RRC sets to accommodate the MCT model.

Suppose that we generate a possible world $g$ from $G$ for $C$ by sampling a node set $T$ for each node $v$ from its triggering distribution $\mathcal{T}_C(v)$ and removing any outgoing edge of $v$ that does not point to a node in $T$. Let $\mathcal{G}$ be the distribution of $g$ induced by the random choices of triggering sets. We refer to $\mathcal{G}$ as the \emph{triggering graph distribution} for campaign $C$ in $G$. Similar to the MCIC setting, a possible world for $L$ under the high effectiveness property is the whole graph $G$. For any given node $v$ and a possible world $g$ for $C$ sampled from $\mathcal{G}$, we define the reverse reachable without cutoff (RRC) set for $v$ in $g$ as the set of nodes that can save $v$ in $g$. In addition, we define a random RRC set as one that is generated on an instance of $g$ randomly sampled from $\mathcal{G}$, for a node selected from $g$ uniformly at random. These random RRC sets are constructed in the same fashion as Algorithm \ref{alg:genRRC} with the modification that the forward randomized BFS samples triggering sets to determine which edges to traverse. By incorporating such an updated forward BFS into Algorithm \ref{alg:nodeselection}, our solution extends to the multi-campaign triggering model. 

Random RRC sets, as defined above, are constructed by a randomized BFS as follows. Let $v$ be a randomly selected node. Given $v$, we first take a sample $T$ from $v$’s triggering distribution $\mathcal{T}_C(v)$, and then put all nodes in $T$ into a queue. Then, we iteratively extract the node at the top of the queue; for each node $u$ extracted, we sample a set $T'$ from $u$'s triggering distribution, and we insert any unvisited node in $T'$ into the queue. When the queue becomes empty, we terminate the process, and form a random RRC set with the nodes visited during the process. The expected cost of the whole process is $O(EPT)$, where $EPT$ denotes the expected number of edges in $G$ that point to the nodes in a random RRC set. This expected time complexity is the same as Algorithm \ref{alg:genRRC} for generating random RRC sets under the MCIC model.


Our next step is to show that the revised solution retains the performance guarantees of \emph{RPS}. We first present an extended version of Lemma \ref{lem:equiv} for the MCT model. (The proof of the lemma is almost identical to that of Lemma \ref{lem:equiv}.)

\begin{lem}
\label{lem:equiv2}
Let $S$ be a fixed set of nodes, $v$ be a fixed node, and $\mathcal{G}$ be the triggering graph distribution for $G$. Suppose that we generate an RRC set $R$ for $v$ on a graph $g$ sampled from $\mathcal{G}$. Let $\rho_1$ be the probability that $S$ overlaps with $R$, and $\rho_2$ be the probability that $S$ (as a seed set for campaign $L$) can save $v$ in an prevention process on $G$ under the MCT model. Then, $\rho_1 = \rho_2$.
\end{lem}

\begin{proof}
Let $S$ be a fixed set of nodes, and $v$ be a fixed node. Suppose that we generate an RRC set $R$ for $v$ on a graph $g \sim \mathcal{G}$. Let $\rho_1$ be the probability that $S$ overlaps with $R$ and let $\rho_2$ be the probability that $S$, when used as a seed set, can save $v$ in a prevention process on $\mathcal{G}$. By the definition of RRC sets under the MCT model, if $v \in R_g(A_C)$ then $\rho_1$ equals the probability that a node $u \in S$ is a saviour for $v$. That is, $\rho_1$ equals the probability that $G$ contains a directed path from $u \in S$ to $v$ and $u \not\in \tau(v)$ and $0$ if $v \not\in R_g(A_C)$. Meanwhile, if $v \in R_g(A_C)$ then $\rho_2$ equals the probability that a node $u \in S$ can save $v$ (i.e. $v \in ( cl_g(u) \setminus obs_g(u) )$) and $0$ if $v \not\in R_g(A_C)$. It follows that $\rho_1 = \rho_2$ due to the equivalence between the set of saviours for $v$ and the ability to save $v$.
\end{proof}

We note that all of the theoretical analysis of \emph{RPS} is based on the Chernoff bounds and Lemma \ref{lem:equiv}, without relying on any other results specific to the MCIC model. Therefore, once we establish an equivalent to Lemma \ref{lem:equiv} (Lemma \ref{lem:equiv2}), it is straightforward to combine it with the Chernoff bounds to show that, under the MCT model, \emph{RPS} provides the same performance guarantees as in the case of the MCIC model. Thus, we have the following theorem:

\begin{thm}
Under MCT, \emph{RPS} runs in $O((k+l)(m+n)(1 / (1 - \gamma)) \log n / \epsilon^2)$ expected time, and returns a $(1 - 1/e - \epsilon)$-approximate solution with at least $1 - 3n^{-l}$ probability.
\end{thm}

\section{Experiments}

The focus of our experiments is \emph{algorithm efficiency} measured in runtime where our goal is to demonstrate the superior performance of \emph{RPS} compared to \emph{MCGreedy}. We observe that \emph{RPS} provides a significant improvement of several orders of magnitude over \emph{MCGreedy}. Further, we confirm that $\frac{1}{1 - \gamma} \ll n+1$ on our small datasets which is strong evidence that \emph{RPS} will outperform \emph{MCGreedy} on typical social networks. Finally, we observe that the vast majority of the computation time is spent on generating the RRC sets for $\mathcal{R}$. A detailed experimental analysis and discussion is presented in this section.

\begin{table}
\caption{Dataset Statistics}
\label{tbl:datasets}
\centering
\begin{tabular}{lrrr} \hline 
Name&           $|V|$& 	        $|E|$&          Average degree \\ \hline \hline
nethept&        15,229& 	    62,752&     4.1\\
word\_assoc&    10,617& 	    72,172&     6.8\\ \hline
dblp-2010&      326,186& 	    1,615,400&  6.1\\
cnr-2000&       325,557& 	    3,216,152&  9.9\\ \hline
ljournal-2008&  5,363,260& 	    79,023,142& 28.5\\ \hline
\end{tabular}
\end{table}

\begin{table}
\caption{$\frac{1}{1 - \gamma}$ values for small datasets.}
\label{tbl:gamma}
\centering
\begin{tabular}{l|cccc} \hline 
\multicolumn{1}{c|}{}&
\multicolumn{2}{c}{word\_assoc}&   
\multicolumn{2}{c}{nethept} \\
k&      top1&   top5&   top1&   top5 \\ \hline
1&      23.4471&      25.9804& 	    48.3194&      48.619 \\
10&     24.8521&      26.6518& 	    60.5&      43.1875 \\
20&     24.7509&      25.2607& 	    57.0111&      61.4167 \\ \hline
\emph{max}&    10,618& 	10,618&  15,230&      15,230 \\ \hline
\end{tabular}
\end{table}

In this section, we present our experimental results. All of our algorithms are implemented in C++ (available at \url{https://github.com/stamps}) and tested on a machine with dual 6 core 2.10GHz Intel Xeon CPUs, 128GB RAM and running Ubuntu 14.04.2.

\vspace{5pt}

\noindent \textbf{Datasets. } The network statistics for all of the datasets we consider are shown in Table~\ref{tbl:datasets}. We obtained the datasets from Laboratory of Web Algorithmics.\footnote{\url{http://law.di.unimi.it/datasets.php}} We divide the datasets by horizontal lines according to their size, small (S), medium (M), and large (L).

\vspace{5pt}

\noindent \textbf{Propagation Model. } We consider the MCIC model (see Section 2.1) of Budak et al. We set the propagation probability of each edge $e$ as follows: we first identify the node $v$ that $e$ points to, and then set $p(e) = 1/i$, where $i$ denotes the in-degree of $v$. This setting of $p(e)$ is widely adopted in prior work \cite{Chen2010, Chen2009EIM, jung2012irie, wang2012scalable}.

\vspace{5pt}

\noindent \textbf{Parameters. } Unless otherwise specified, we set $\epsilon = 0.1$ in our experiments. We set $l$ in a way that ensures a success probability of $1 - 1/n$. For \emph{MCGreedy}, we set the number of Monte Carlo steps to $r = 10000$, following the standard practice in the literature. Note that this choice of $r$ is to the advantage of \emph{MCGreedy} because the value of $r$ required to achieve the same theoretical guarantees as \emph{RPS} in our experiments is always much larger than $10000$. In each experiment, we repeat each run five times and report the average result.

We are interested in simulating the misinformation prevention process when the bad campaign $C$ has a sizable influence on the network to best demonstrate how the techniques could be used in real world settings. That is, we believe the scenario in which we are attempting to prevent the spread of misinformation when the bad campaign has the ability to influence a large fraction of the network to be more relevant than when only a very small number of users would adopt the bad campaign. Towards this end, we first compute the \emph{top-k} influential vertices for each network and then randomly select the seed set $A_C$ from the \emph{top-1} and \emph{top-5} vertices for each experiment. This process ensures the misinformation has a large potential influence in the network.

The focus of our experiments is \emph{algorithm efficiency} measured in runtime where our goal is to demonstrate the superior performance of \emph{RPS} compared to \emph{MCGreedy}. Meanwhile, we observed that the \emph{algorithm accuracy} (measured in percentage of nodes saved) of \emph{RPS} matches \emph{MCGreedy} very closely. We observe that, consistent with the results reported in \cite{budak2011limiting}, \emph{RPS} quickly approaches a maximal expected prevention value as $k$ increases across all datasets. This is natural since both \emph{RPS} and \emph{MCGreedy} are maximizing a submodular objective function in a greedy fashion. The novelty of \emph{RPS} addresses the scalability hurdle in a similar sense to Borgs et. al.\ \cite{borgs2012} in relation to Kempe et. al.\ \cite{kempe2003}.  For a detailed comparison of the accuracy of \emph{MCGreedy} compared to a number of natural heuristics we refer the interested reader to \cite{budak2011limiting}.

\vspace{5pt}

\noindent \textbf{Plots.} First, we plot the runtimes of \emph{MCGreedy} and \emph{RPS} for a single seed in Figure \ref{fig:greedy_comparison} and observe that \emph{RPS} provides a significant improvement of several orders of magnitude over \emph{MCGreedy}. Note, we only compare \emph{RPS} to \emph{MCGreedy} on the smallest networks due to the substantial runtime required for \emph{MCGreedy}. Furthermore, for similar scaling issues of \emph{MCGreedy}, we restrict our comparison to $k = 1$. However, since both approaches scale linearly with $k$ we can conclude that \emph{RPS} offers a tremendous runtime improvement over the approach of \cite{budak2011limiting}.

Next, we show the total runtimes (Figures \ref{fig:small_time}, \ref{fig:med_time}) and computation breakdowns (Figures \ref{fig:small_time_breakdown}, \ref{fig:ljournal_time_breakdown}, \ref{fig:med_time_breakdown}) for each dataset. We observe that the vast majority of the computation time is spent on generating the RRC sets for $\mathcal{R}$. Furthermore, the amount of time spent on computing a lower bound increases across all datasets though remains a small fraction of the overall runtime. As expected, the time spent refining the lower bound estimate remains a very small fraction of total computation time due to the small number of iterations of the algorithm that improves the lower bound estimate and only takes up a relatively large fraction of the computation time on the cnr-2000 top5 dataset (Figure \ref{fig:med_time_breakdown}c). The density of the cnr-2000 network leads to larger RRC sets that results in a larger fraction of time spent on computing and refining a lower bound. Finally, we plot the memory consumption statistics in Figure \ref{fig:small_med_mem}.

\vspace{5pt}

\noindent \textbf{Running-time Results.} We compare the runtime trends of our results for the EIL problem to those of \cite{Tang2014} for the IM problem. Tang et al.\ report that, when $k$ increases, the runtime of their approach (\emph{TIM}) tends to decrease before eventually increasing. They explain this by considering the breakdown of the computation times required by each algorithm in \emph{TIM}. They observe that the computation time is mainly incurred by their analog to Algorithm 1 (the node selection phase) which is primarily determined by the number $\theta$ of RR sets that need to be generated. They have $\theta = \lambda/KPT^{+}$, where $\lambda$ is analogous to ours, and $KPT^{+}$ is a lower-bound on the optimal influence of a size-$k$ node set. In both the IC and MCIC models, the analogs of $\lambda$ and $KPT^{+}$ increase with $k$, and it happens that for the IM problem, Tang et al.\ observe that the increase of $KPT^{+}$ is more pronounced than that of $\lambda$ for smaller values of $k$, which leads to the decrease in \emph{TIM}’s runtime. 

On the contrary, for the EIL problem, the increase of $KPT^{+}$ does not dominate to a point that the runtime of \emph{RPS} decreases as $k$ increases. Instead, we see a linear increase in runtime as $k$ increases for all the networks considered. To explain, consider how $KPT^{+}$ grows in each setting. In the MCIC model we see that $KPT^{+}$ rapidly approaches its maximal value which corresponds to the growth of $KPT^{+}$ plateauing much sooner. In contrast, in the IC model, the analogous $KPT^{+}$ value continues to grow at a significant rate for a wider range of $k$ values since the ceiling for the maximal influence is not tied to a second campaign, as it is in the MCIC model. As such, the influence estimates do not level off as quickly. This translates to the growth of $KPT^{+}$ outpacing the growth of $\lambda$.

\vspace{5pt}

\noindent \textbf{Memory Consumption. } Another set of experiments monitors the memory consumption required to store the RRC set structure $\mathcal{R}$. We observe that the size of $\mathcal{R}$ for the EIL problem is larger than that required by the IM problem. Using the ``hypergraph'' nomenclature due to Borgs et al.\ $\mathcal{R}$ is viewed as a hypergraph with each RRC set in $\mathcal{R}$ corresponding to a hyperedge. We observe that the hyperedges generated for the IM problem are non-empty in every iteration of the algorithm. Additionally, each hyperedge has relatively small size. The result is that the hypergraph generated for the IM problem is very dense, but each hyperedge is relatively ``light'' (i.e. it contains few nodes).

In contrast, in each iteration of \emph{RPS} we have a substantial probability to produce an empty RRC set, since we require that a randomly selected node is in the randomized BFS tree resulting from the influence propagation process initialized at $A_C$. These empty RRC sets are necessary for the computation of expected prevention to be accurate, but results in a hypergraph that differs significantly in structure from those of the IM problem.

In particular, since the \emph{generateRRC} algorithm is a deterministic BFS (with specialized stopping conditions to account for cutoff nodes) it reaches a much larger fraction of the network. Therefore, while there are far fewer non-empty hyperedges generated, they are much large in size: often on the order of half the network. Thus, the resulting hypergraph is sparse, but contains very ``heavy'' edges. These two opposing metrics, a dense hypergraph with ``light'' hyperedges versus a sparse hypergraph with ``heavy'' hyperedges, result in the latter requiring more memory to store. Despite a larger memory requirement compared to the single campaign setting we show that our approach has the ability to scale far beyond what was achieved by Budak et al.\ and provides orders of magnitude improvement for the runtime.

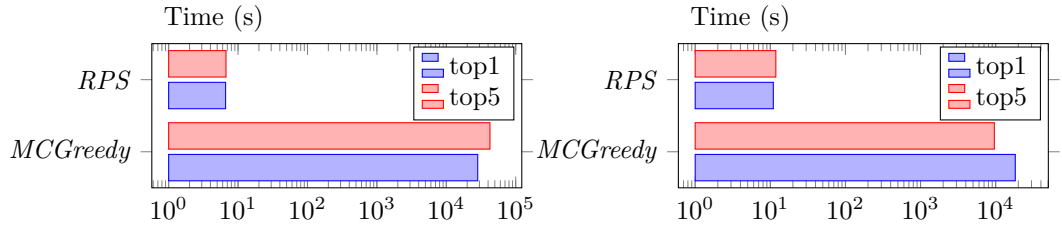
\begin{figure}
\centering
\begin{subfigure}{.5\linewidth}
\centering
\begin{tikzpicture}
\begin{semilogxaxis}[
    xbar, xmin=1, xmax=70000,
    width=6.5cm, height=3.5cm, enlarge y limits=0.5, enlarge x limits=0.05,
    xlabel={Time (s)},
    symbolic y coords={\emph{MCGreedy},\emph{RPS}},
    ytick=data,
    every axis x label/.style={at={(current axis.north west)},above right=0.5mm},
    nodes near coords align={horizontal}]
\addplot coordinates {(28366.7,\emph{MCGreedy}) (6.619217,\emph{RPS})};
\addplot coordinates {(42557.2,\emph{MCGreedy}) (6.70506,\emph{RPS})};
\legend{top1,top5}
\end{semilogxaxis}
\end{tikzpicture}
\end{subfigure}%
\begin{subfigure}{.5\linewidth}
\centering
\begin{tikzpicture}
\begin{semilogxaxis}[
    xbar, xmin=1, xmax=30000,
    width=6.5cm, height=3.5cm, enlarge y limits=0.5, enlarge x limits=0.05,
    xlabel={Time (s)},
    symbolic y coords={\emph{MCGreedy},\emph{RPS}},
    ytick=data,
    every axis x label/.style={at={(current axis.north west)},above right=0.5mm},
    nodes near coords align={horizontal}]
\addplot coordinates {(10.964486,\emph{RPS}) (18237.7,\emph{MCGreedy})};
\addplot coordinates {(11.855423,\emph{RPS}) (9594.78,\emph{MCGreedy})};
\legend{top1,top5}
\end{semilogxaxis}
\end{tikzpicture}
\end{subfigure}%
\caption{Runtimes comparison between \emph{RPS} and \emph{MCGreedy} for wordassociation-2011 (left) and nethept (right) datasets. \label{fig:greedy_comparison}}
\end{figure}

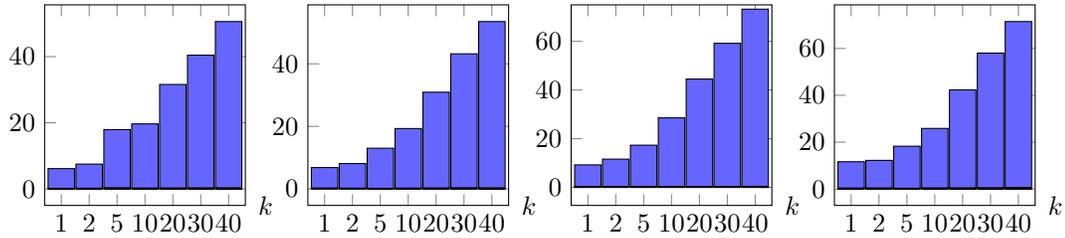
\begin{figure}
\centering
\begin{subfigure}{.25\linewidth}
\centering
\begin{tikzpicture}
	\begin{axis}[
	width=4.25cm, height=4.25cm,
	legend pos=north west,
	legend style={font=\tiny},
	ybar stacked, 
	xtick=data,
	xlabel=$k$,
	every axis x label/.style={at={(current axis.south east)},right=0.5mm},
	every axis y label/.style={at={(current axis.north west)},above right=0.5mm},
	symbolic x coords={1,2,5,10,20,30,40}
	]
	\addplot[fill=green] coordinates
		{(1,0.10154052) (2,0.11900546) (5,0.14432278) (10,0.11257432) (20,0.10763686) (30,0.11594178) (40,0.09383006)};
	\addplot[fill=red] coordinates
		{(1,0.10563084) (2,0.10978426) (5,0.17384058) (10,0.14736236) (20,0.16580838) (30,0.18340876) (40,0.21169554)};
	\addplot[fill=blue!60] coordinates
		{(1,5.957708) (2,7.298486) (5,17.62837) (10,19.41752) (20,31.310244) (30,40.138534) (40,50.281834)};
	\end{axis}
\end{tikzpicture}
\end{subfigure}%
\begin{subfigure}{.25\linewidth}
\centering
\begin{tikzpicture}
	\begin{axis}[
	width=4.25cm, height=4.25cm,
	legend pos=north west,
	legend style={font=\tiny},
	ybar stacked, 
	xtick=data,
	xlabel=$k$,
	every axis x label/.style={at={(current axis.south east)},right=0.5mm},
	every axis y label/.style={at={(current axis.north west)},above right=0.5mm},
	symbolic x coords={1,2,5,10,20,30,40}
	]
	\addplot[fill=green] coordinates
		{(1,0.13133926) (2,0.11201688) (5,0.06937584) (10,0.08819788) (20,0.09644416) (30,0.10176842) (40,0.07523506)};
	\addplot[fill=red] coordinates
		{(1,0.10896942) (2,0.10720152) (5,0.1206947) (10,0.1440813) (20,0.16294378) (30,0.19363872) (40,0.20372228)};
	\addplot[fill=blue!60] coordinates
		{(1,6.593928) (2,7.882254) (5,12.827928) (10,19.077142) (20,30.658168) (30,42.904254) (40,53.22744)};
	\end{axis}
\end{tikzpicture}
\end{subfigure}%
\begin{subfigure}{.25\linewidth}
\centering
\begin{tikzpicture}
	\begin{axis}[
	ymax=75,
	width=4.25cm, height=4.25cm,
	legend pos=north west,
	legend style={font=\tiny},
	ybar stacked, 
	xtick=data,
	xlabel=$k$,
	every axis x label/.style={at={(current axis.south east)},right=0.5mm},
	every axis y label/.style={at={(current axis.north west)},above right=0.5mm},
	symbolic x coords={1,2,5,10,20,30,40}
	]
	\addplot[fill=green] coordinates
		{(1,0.2496012) (2,0.210817) (5,0.1227278) (10,0.1561882) (20,0.11058208) (30,0.1246392) (40,0.1627258)};
	\addplot[fill=red] coordinates
		{(1,0.16798664) (2,0.1570053) (5,0.1668211) (10,0.20812752) (20,0.23561366) (30,0.2601144) (40,0.2962596)};
	\addplot[fill=blue!60] coordinates
		{(1,8.860476) (2,11.273142) (5,17.07482) (10,28.225528) (20,44.175866) (30,58.82455) (40,72.77688)};
	\end{axis}
\end{tikzpicture}
\end{subfigure}%
\begin{subfigure}{.25\linewidth}
\centering
\begin{tikzpicture}
	\begin{axis}[
	width=4.25cm, height=4.25cm,
	legend pos=north west,
	legend style={font=\tiny},
	ybar stacked, 
	xtick=data,
	xlabel=$k$,
	every axis x label/.style={at={(current axis.south east)},right=0.5mm},
	every axis y label/.style={at={(current axis.north west)},above right=0.5mm},
	symbolic x coords={1,2,5,10,20,30,40}
	]
	\addplot[fill=green] coordinates
		{(1,0.1702292) (2,0.1616758) (5,0.135761) (10,0.1820038) (20,0.1540508) (30,0.1857244) (40,0.2135116)};
	\addplot[fill=red] coordinates
		{(1,0.2013276) (2,0.2297169) (5,0.18481668) (10,0.2203396) (20,0.26765) (30,0.2808096) (40,0.3449596)};
	\addplot[fill=blue!60] coordinates
		{(1,11.277608) (2,11.840194) (5,17.953984) (10,25.46388) (20,41.905752) (30,57.560964) (40,70.91841)};
	\end{axis}
\end{tikzpicture}
\end{subfigure}%
\caption{Breakdown of computation time (s) for small datasets. Blue stack corresponds to Algorithm \ref{alg:nodeselection}, red to improving the lower bound estimation, and green (which is almost invisible) to computing the initial lower bound estimate. Listed left to right: word\_assoc top1, word\_assoc top5, nethept top1, nethept top5. \label{fig:small_time_breakdown}}
\end{figure}

\begin{figure}
\centering
\begin{subfigure}{.3\linewidth}
\centering
\begin{tikzpicture}
    \begin{axis}[
        width=4cm, height=4cm,
        xlabel=$k$,
        every axis x label/.style={at={(current axis.south east)},right=0.5mm},
        every axis y label/.style={at={(current axis.north west)},above right=0.5mm}
    ]
    \addplot[mark=*,blue] plot coordinates {
        (1,6.09105)
        (2,7.45459)
        (5,12.7498)
        (10,19.6026)
        (20,31.5223)
        (30,40.3709)
        (40,50.5107)
    };
    \addplot[mark=*,red] plot coordinates {
        (1,6.77399)
        (2,8.02968)
        (5,12.9438)
        (10,19.2366)
        (20,30.8592)
        (30,43.1432)
        (40,53.4381)
    };
    \end{axis}
\end{tikzpicture}
\end{subfigure}%
\begin{subfigure}{.3\linewidth}
\centering
\begin{tikzpicture}
    \begin{axis}[
        width=4cm, height=4cm,
        xlabel=$k$,
        every axis x label/.style={at={(current axis.south east)},right=0.5mm},
        every axis y label/.style={at={(current axis.north west)},above right=0.5mm}
    ]
    \addplot[mark=*,blue] plot coordinates {
        (1,9.19081)
        (2,11.5597)
        (5,17.2726)
        (10,28.4985)
        (20,44.424)
        (30,59.1056)
        (40,73.1218)
    };
    \addplot[mark=*,red] plot coordinates {
        (1,11.5619)
        (2,12.1272)
        (5,18.1834)
        (10,25.7845)
        (20,42.2212)
        (30,57.9243)
        (40,71.378)
    };
    \end{axis}
\end{tikzpicture}
\end{subfigure}%
\caption{Runtimes (s) for small datasets. word\_assoc on the left and nethept on the right with blue for top1 and red for top5. \label{fig:small_time}}
\end{figure}
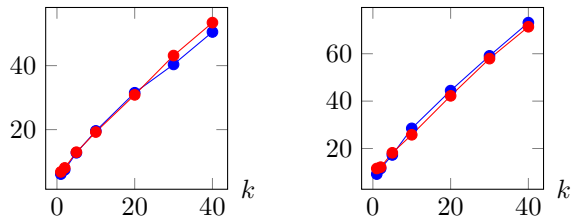

\begin{figure}
\centering
\begin{subfigure}{.3\linewidth}
\centering
\begin{tikzpicture}
	\begin{axis}[
	width=4.5cm, height=4.5cm,
	legend pos=north west,
	legend style={font=\tiny},
	ybar stacked, 
	xtick=data,
	xlabel=$k$,
	every axis x label/.style={at={(current axis.south east)},right=0.5mm},
	every axis y label/.style={at={(current axis.north)},above right=0.5mm},
	symbolic x coords={1,2,3,5,7,9,10}
	]
	\addplot[fill=green] coordinates
		{(1,184.7496) (2,185.802) (3,125.0443) (5,120.322) (7,184.654) (9,119.713) (10,144.9586)};
	\addplot[fill=red] coordinates
		{(1,61.27783) (2,73.37079) (3,82.60082) (5,106.7732) (7,118.3351) (9,139.8171) (10,165.19884)};
	\addplot[fill=blue!60] coordinates
		{(1,16401.75866) (2,20385.38) (3,25876.4666) (5,34410.08) (7,41319.68) (9,51909.88) (10,53080.17)};
	\end{axis}
\end{tikzpicture}
\end{subfigure}%
\begin{subfigure}{.3\linewidth}
\centering
\begin{tikzpicture}
	\begin{axis}[
	width=4.5cm, height=4.5cm,
	legend pos=north west,
	legend style={font=\tiny},
	ybar stacked, 
	xtick=data,
	xlabel=$k$,
	every axis x label/.style={at={(current axis.south east)},right=0.5mm},
	every axis y label/.style={at={(current axis.north)},above right=0.5mm},
	symbolic x coords={1,2,3,5,7,9,10}
	]
	\addplot[fill=green] coordinates
		{(1,150.042) (2,127.726) (3,146.322) (5,149.431) (7,138.265) (9,141.705) (10,145.854)};
	\addplot[fill=red] coordinates
		{(1,61.29615) (2,89.99468) (3,104.90273) (5,115.70864) (7,151.24381) (9,146.53681) (10,172.8951)};
	\addplot[fill=blue!60] coordinates
		{(1,16798.983) (2,23245.52) (3,30576.6) (5,38463.35) (7,53742.4) (9,59308.59) (10,62090.74)};
	\end{axis}
\end{tikzpicture}
\end{subfigure}%
\caption{Breakdown of computation time (s) for ljournal-2008. Blue stack corresponds to Algorithm \ref{alg:nodeselection}, red to improving the lower bound estimation, and green (which is almost invisible) to computing the initial lower bound estimate. Left is top1 and right is top5. \label{fig:ljournal_time_breakdown}}
\end{figure}
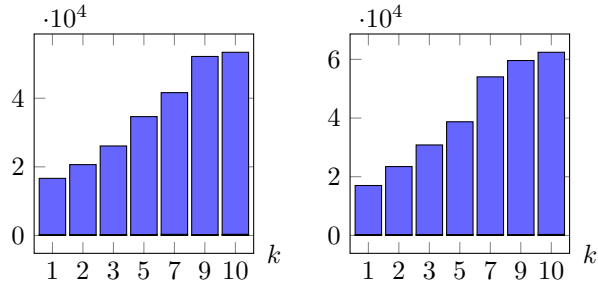

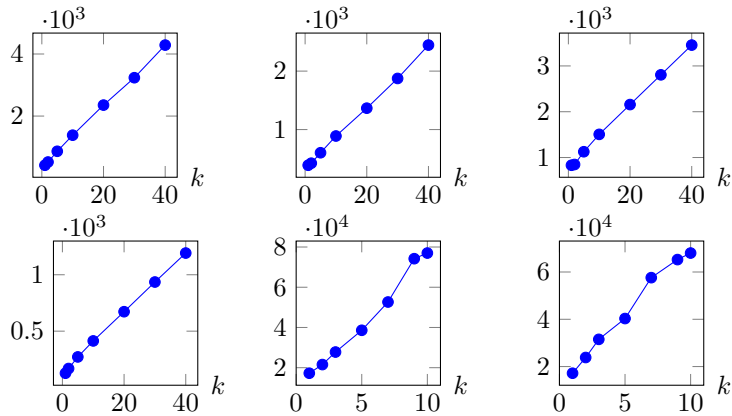
\begin{figure}
\centering
\begin{subfigure}{.25\linewidth}
\begin{tikzpicture}
    \begin{axis}[
        width=3.5cm, height=3.5cm,
        xlabel=$k$,
        scaled y ticks=base 10:-3,
        every axis x label/.style={at={(current axis.south east)},right=0.5mm},
        every axis y label/.style={at={(current axis.north west)},above right=0.5mm}
    ]
    \addplot[mark=*,blue] plot coordinates {
        (1,415.283)
        (2,519.109)
        (5,865.174)
        (10,1384.01)
        (20,2353.37)
        (30,3235.82)
        (40,4287.17)
    };
    \end{axis}
\end{tikzpicture}
\end{subfigure}%
\begin{subfigure}{.25\linewidth}
\begin{tikzpicture}
    \begin{axis}[
        width=3.5cm, height=3.5cm,
        xlabel=$k$,
        scaled y ticks=base 10:-3,
        every axis x label/.style={at={(current axis.south east)},right=0.5mm},
        every axis y label/.style={at={(current axis.north west)},above right=0.5mm}
    ]
    \addplot[mark=*,blue] plot coordinates {
        (1,389.673)
        (2,424.891)
        (5,603.546)
        (10,889.725)
        (20,1365.87)
        (30,1874.33)
        (40,2445.04)
    };
    \end{axis}
\end{tikzpicture}
\end{subfigure}%
\begin{subfigure}{.25\linewidth}
\begin{tikzpicture}
    \begin{axis}[
        width=3.5cm, height=3.5cm,
        xlabel=$k$,
        scaled y ticks=base 10:-3,
        every axis x label/.style={at={(current axis.south east)},right=0.5mm},
        every axis y label/.style={at={(current axis.north west)},above right=0.5mm}
    ]
    \addplot[mark=*,blue] plot coordinates {
        (1,831.133)
        (2,847.479)
        (5,1126.15)
        (10,1505.47)
        (20,2156.14)
        (30,2805.46)
        (40,3455.06)
    };
    \end{axis}
\end{tikzpicture}
\end{subfigure}%

\begin{subfigure}{.25\linewidth}
\begin{tikzpicture}
    \begin{axis}[
        width=3.5cm, height=3.5cm,
        xlabel=$k$,
        scaled y ticks=base 10:-3,
        every axis x label/.style={at={(current axis.south east)},right=0.5mm},
        every axis y label/.style={at={(current axis.north west)},above right=0.5mm}
    ]
    \addplot[mark=*,blue] plot coordinates {
        (1,127.12)
        (2,168.632)
        (5,271.544)
        (10,413.409)
        (20,672.362)
        (30,935.525)
        (40,1192.32)
    };
    \end{axis}
\end{tikzpicture}
\end{subfigure}%
\begin{subfigure}{.25\linewidth}
\begin{tikzpicture}
    \begin{axis}[
        width=3.5cm, height=3.5cm,
        xlabel=$k$,
        xmin=0,
        every axis x label/.style={at={(current axis.south east)},right=0.5mm},
        every axis y label/.style={at={(current axis.north)},above right=0.5mm}
    ]
    \addplot[mark=*,blue] plot coordinates {
        (1,17253.5)
        (2,21553.1)
        (3,27780.9)
        (5,38598.6)
        (7,52654)
        (9,74192.2)
        (10,76994.16)
    };
    \end{axis}
\end{tikzpicture}
\end{subfigure}%
\begin{subfigure}{.25\linewidth}
\begin{tikzpicture}
    \begin{axis}[
        width=3.5cm, height=3.5cm,
        xlabel=$k$,
        xmin=0,
        every axis x label/.style={at={(current axis.south east)},right=0.5mm},
        every axis y label/.style={at={(current axis.north)},above right=0.5mm}
    ]
    \addplot[mark=*,blue] plot coordinates {
        (1,17183.3)
        (2,23830)
        (3,31511.6)
        (5,40308.7)
        (7,57607.8)
        (9,65249)
        (10,68014)
    };
    \end{axis}
\end{tikzpicture}
\end{subfigure}%
\caption{Runtimes (s) for medium \& large datasets. Listed left to right: dblp top1, dblp top5, cnr top1, cnr top5, ljournal-2008 top1, ljournal-2008 top5. \label{fig:med_time}}
\end{figure}

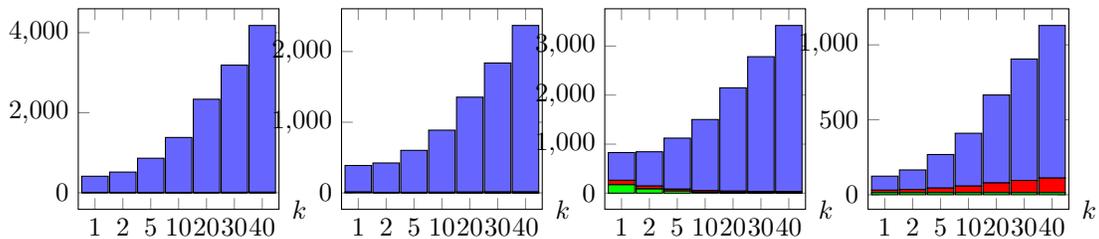
\begin{figure}
\centering
\begin{subfigure}{.25\linewidth}
\begin{tikzpicture}
	\begin{axis}[
	width=4.25cm, height=4.25cm,
	legend pos=north west,
	legend style={font=\tiny},
	ybar stacked, 
	xtick=data,
	xlabel=$k$,
	every axis x label/.style={at={(current axis.south east)},right=0.5mm},
	every axis y label/.style={at={(current axis.north west)},above right=0.5mm},
	symbolic x coords={1,2,5,10,20,30,40}
	]
	\addplot[fill=green] coordinates
		{(1,7.94428) (2,6.543536) (5,5.021438) (10,4.433742) (20,4.258454) (30,4.352406) (40,4.302964)};
	\addplot[fill=red] coordinates
		{(1,3.029352) (2,2.986437) (5,3.9417574) (10,4.773761) (20,6.7428282) (30,7.9309048) (40,9.9757478)};
	\addplot[fill=blue!60] coordinates
		{(1,402.87862) (2,508.13108) (5,854.7302) (10,1372.7451) (20,2326.3304) (30,3176.0656) (40,4165.3668)};
	\end{axis}
\end{tikzpicture}
\end{subfigure}%
\begin{subfigure}{.25\linewidth}
\begin{tikzpicture}
	\begin{axis}[
	width=4.25cm, height=4.25cm,
	legend pos=north west,
	legend style={font=\tiny},
	ybar stacked, 
	xtick=data,
	xlabel=$k$,
	every axis x label/.style={at={(current axis.south east)},right=0.5mm},
	every axis y label/.style={at={(current axis.north west)},above right=0.5mm},
	symbolic x coords={1,2,5,10,20,30,40}
	]
	\addplot[fill=green] coordinates
		{(1,8.185192) (2,5.150426) (5,4.403582) (10,4.393818) (20,4.302406) (30,4.406432) (40,4.161)};
	\addplot[fill=red] coordinates
		{(1,3.537287) (2,3.068092) (5,4.1991762) (10,5.3121472) (20,7.0291446) (30,8.709089) (40,9.4447158)};
	\addplot[fill=blue!60] coordinates
		{(1,376.45626) (2,415.171) (5,593.40394) (10,878.0222) (20,1343.0562) (30,1824.9268) (40,2355.0002)};
	\end{axis}
\end{tikzpicture}
\end{subfigure}%
\begin{subfigure}{.25\linewidth}
\begin{tikzpicture}
	\begin{axis}[
	width=4.25cm, height=4.25cm,
	legend pos=north west,
	legend style={font=\tiny},
	ybar stacked, 
	xtick=data,
	xlabel=$k$,
	every axis x label/.style={at={(current axis.south east)},right=0.5mm},
	every axis y label/.style={at={(current axis.north west)},above right=0.5mm},
	symbolic x coords={1,2,5,10,20,30,40}
	]
	\addplot[fill=green] coordinates
		{(1,171.8056) (2,86.07618) (5,43.45636) (10,21.8096) (20,15.39546) (30,11.09) (40,11.1384)};
	\addplot[fill=red] coordinates
		{(1,89.85303) (2,58.495158) (5,38.530678) (10,30.51092) (20,25.75966) (30,25.837188) (40,24.17963)};
	\addplot[fill=blue!60] coordinates
		{(1,566.35536) (2,699.85542) (5,1041.10178) (10,1449.83114) (20,2107.3028) (30,2746.5882) (40,3384.4528)};
	\end{axis}
\end{tikzpicture}
\end{subfigure}%
\begin{subfigure}{.25\linewidth}
\begin{tikzpicture}
	\begin{axis}[
	width=4.25cm, height=4.25cm,
	legend pos=north west,
	legend style={font=\tiny},
	ybar stacked, 
	xtick=data,
	xlabel=$k$,
	every axis x label/.style={at={(current axis.south east)},right=0.5mm},
	every axis y label/.style={at={(current axis.north west)},above right=0.5mm},
	symbolic x coords={1,2,5,10,20,30,40}
	]
	\addplot[fill=green] coordinates
		{(1,14.1815) (2,14.15882) (5,14.18462) (10,14.1687) (20,14.2806) (30,14.17542) (40,14.21618)};
	\addplot[fill=red] coordinates
		{(1,16.011972) (2,20.111222) (5,30.109782) (10,43.516602) (20,64.43484) (30,81.20456) (40,96.94218)};
	\addplot[fill=blue!60] coordinates
		{(1,93.8155) (2,131.27174) (5,224.12474) (10,352.4913) (20,586.8226) (30,810.297) (40,1018.2152)};
	\end{axis}
\end{tikzpicture}
\end{subfigure}
\caption{Breakdown of computation time (s) for medium datasets. Blue stack corresponds to Algorithm \ref{alg:nodeselection}, red to improving the lower bound estimation, and green (which is almost invisible) to computing the initial lower bound estimate. Listed left to right: dblp top1, dblp top5, cnr top1, cnr top5. \label{fig:med_time_breakdown}}
\end{figure}

\begin{figure}
\centering
\begin{subfigure}{.25\linewidth}
\begin{tikzpicture}
    \begin{axis}[
        legend pos=south east,
        width=3.75cm, height=3.75cm,
        xlabel=$k$,
        every axis x label/.style={at={(current axis.south east)},right=0.5mm},
        every axis y label/.style={at={(current axis.north west)},above right=0.5mm}
    ]
    \addplot[mark=*,blue] plot coordinates {
        (1,0.04407912975)
        (2,0.05403362895)
        (5,0.131181263175)
        (10,0.14628792705)
        (20,0.235589372025)
        (30,0.300812257275)
        (40,0.380422770075)
    };
    \addplot[mark=*,red] plot coordinates {
        (1,0.04786621365)
        (2,0.05662417695)
        (5,0.091813427175)
        (10,0.13711271565)
        (20,0.2205692367)
        (30,0.31021865385)
        (40,0.38529358425)
    };
    \end{axis}
\end{tikzpicture}
\end{subfigure}%
\begin{subfigure}{.25\linewidth}
\begin{tikzpicture}
    \begin{axis}[
        legend pos=south east,
        width=3.75cm, height=3.75cm,
        xlabel=$k$,
        every axis x label/.style={at={(current axis.south east)},right=0.5mm},
        every axis y label/.style={at={(current axis.north west)},above right=0.5mm}
    ]
    \addplot[mark=*,blue] plot coordinates {
        (1,0.07012726065)
        (2,0.088900429675)
        (5,0.133805420525)
        (10,0.22281016885)
        (20,0.3488430885)
        (30,0.465675027075)
        (40,0.57654024345)
    };
    \addplot[mark=*,red] plot coordinates {
        (1,0.095937750575)
        (2,0.09771763995)
        (5,0.146447013425)
        (10,0.20861826375)
        (20,0.342994010325)
        (30,0.473710608925)
        (40,0.587839780725)
    };
    \end{axis}
\end{tikzpicture}
\end{subfigure}%
\begin{subfigure}{.25\linewidth}
\begin{tikzpicture}
    \begin{axis}[
        legend pos=south east,
        width=3.75cm, height=3.75cm,
        xlabel=$k$,
        every axis x label/.style={at={(current axis.south east)},right=0.5mm},
        every axis y label/.style={at={(current axis.north west)},above right =0.5mm}
    ]
    \addplot[mark=*,blue] plot coordinates {
        (1,1.526776)
        (2,1.882301)
        (5,2.7808671)
        (10,3.848836)
        (20,5.574980)
        (30,7.27829)
        (40,8.93966)
    };
    \addplot[mark=*,red] plot coordinates {
        (1,0.801065)
        (2,1.115871)
        (5,1.921811)
        (10,3.08979)
        (20,5.19601)
        (30,7.1556411)
        (40,9.003388)
    };
    \end{axis}
\end{tikzpicture}
\end{subfigure}%

\begin{subfigure}{.25\linewidth}
\begin{tikzpicture}
    \begin{axis}[
        legend pos=south east,
        width=3.75cm, height=3.75cm,
        xlabel=$k$,
        every axis x label/.style={at={(current axis.south east)},right=0.5mm},
        every axis y label/.style={at={(current axis.north west)},above right=0.5mm}
    ]
    \addplot[mark=*,blue] plot coordinates {
        (1,0.887221)
        (2,1.124224)
        (5,1.877000)
        (10,3.032771)
        (20,5.119721)
        (30,6.987006)
        (40,9.170743)
    };
    \addplot[mark=*,red] plot coordinates {
        (1,1.501678)
        (2,1.658125)
        (5,2.335703)
        (10,3.420936)
        (20,5.214368)
        (30,7.1000865)
        (40,9.130239)
    };
    \end{axis}
\end{tikzpicture}
\end{subfigure}%
\begin{subfigure}{.25\linewidth}
\begin{tikzpicture}
    \begin{axis}[
        legend pos=south east,
        width=3.75cm, height=3.75cm,
        xlabel=$k$,
        ymax=150,
        xmin=0,
        every axis x label/.style={at={(current axis.south east)},right=0.5mm},
        every axis y label/.style={at={(current axis.north west)},above right=0.5mm}
    ]
    \addplot[mark=*,blue] plot coordinates {
        (1,44.22968)
        (2,54.8534)
        (3,69.40997)
        (5,91.70146)
        (7,108.58634)
        (9,129.6161)
        (10,136.0799)
    };
    \addplot[mark=*,red] plot coordinates {
        (1,32.9478)
        (2,45.48916)
        (3,59.68772)
        (5,74.9220)
        (7,104.0023)
        (9,114.269)
        (10,119.5658)
    };
    \end{axis}
\end{tikzpicture}
\end{subfigure}%
\caption{Memory consumption (Gb) for all datasets. Blue corresponds to top1 and red to top5. Listed left to right: word\_assoc, nethept, cnr, dblp \& ljournal-2008. \label{fig:small_med_mem}}
\end{figure}
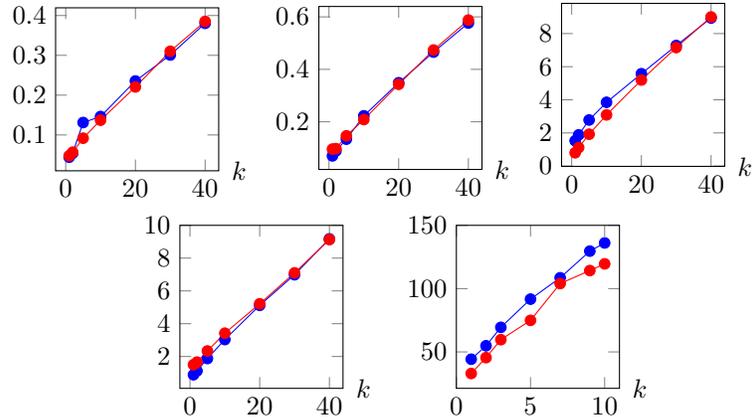

\section{Conclusion \& Future Work}

In this work we presented \emph{RPS}, a novel and scalable approach to the EIL problem. We showed the correctness and a detailed running-time analysis of our approach. Furthermore, we provided two lower bound results: one on the running-time requirement for any approach to solve the EIL problem and another on the number of Monte Carlo simulations required by \emph{MCGreedy} to return a correct solution with high probability. As a result, the expected runtime of \emph{RPS} is always less than the expected runtime of \emph{MCGreedy}. Finally, we describe how our approach can be generalized to a multi-campaign triggering model. In future work we plan to investigate how to adapt our approach to a scenario where the source of the misinformation is only partially known.

\balance

\bibliographystyle{plainurl}
\bibliography{misinformation_prevention_arxiv_full}

\end{document}